\documentclass[submission,copyright,creativecommons,a4paper]{eptcs}
\providecommand{\event}{QPL 2019} 
\usepackage{breakurl}             
\usepackage{underscore}           

\usepackage{hyperref}
\hypersetup{
    colorlinks,
    linkcolor={red!50!black},
    citecolor={blue!50!black},
    urlcolor={blue!80!black}
}

\newcommand{\unbig}[1]{#1}
\usepackage{tikzit}

\usepackage{amsmath,amsthm,amssymb}
\usepackage{stmaryrd,cmll}

\usepackage{xspace,enumerate,color,epsfig}
\usepackage{graphicx}
\graphicspath{{.}{./figures/}}
\usepackage{marvosym}
\usepackage{bm}

\usepackage{docmute}
\usepackage{keycommand}

\newkeycommand{\boxmap}[style=small box][1]{\,\tikz{\node[style=\commandkey{style}] (x) {$#1$};\draw(0,-1.25)--(x)--(0,1.25);}\,}
\newkeycommand{\dboxmap}[style=dbox][1]{\,\tikz{\node[style=\commandkey{style}] (x) {$#1$};\draw[boldedge] (0,-1.25)--(x)--(0,1.25);}\,}

\newkeycommand{\wirelabel}[style=white label]{\,\tikz{\node[style=\commandkey{style}]{};}\,\xspace}

\newkeycommand{\pointmap}[style=point][1]{\,\tikz{\node[style=\commandkey{style}] (x) at (0,-0.3) {$#1$};\node [style=none] (3) at (0, 0.9) {};\node [style=none] (3) at (0, -0.9) {};\draw (x)--(0,0.7);}\,}

\newkeycommand{\point}[style=point,estyle=][1]{\,\tikz{\node[style=\commandkey{style}] (x) at (0,-0.05) {$#1$};\draw [\commandkey{estyle}] (x)--(0,1.05);}\,}

\newkeycommand{\copointmap}[style=copoint][1]{\,\tikz{\node[style=\commandkey{style}] (x) at (0,0.3) {$#1$};\draw (x)--(0,-0.7);}\,}

\newkeycommand{\dpoint}[style=point][1]{\,\tikz{\node[style=\commandkey{style},doubled] (x) at (0,-0.3) {$#1$};\draw[boldedge] (x)--(0,0.7);}\,}

\newkeycommand{\kpoint}[style=kpoint,estyle=][1]{\,\tikz{\node[style=\commandkey{style}] (x) at (0,-0.05) {$#1$};\draw [\commandkey{estyle}] (x)--(0,1.05);}\,}

\newkeycommand{\kpointgray}[style=gray kpoint,estyle=][1]{\,\tikz{\node[style=\commandkey{style}] (x) at (0,-0.05) {$#1$};\draw [\commandkey{estyle}] (x)--(0,1.05);}\,}

\newkeycommand{\typedkpoint}[style=kpoint,edgestyle=][2]{%
\,\begin{tikzpicture}
  \begin{pgfonlayer}{nodelayer}
    \node [style=none] (0) at (0, 1) {};
    \node [style=\commandkey{style}] (1) at (0, -0.75) {$#2$};
    \node [style=right label] (2) at (0.25, 0.5) {$#1$};
  \end{pgfonlayer}
  \begin{pgfonlayer}{edgelayer}
    \draw [\commandkey{edgestyle}] (1) to (0.center);
  \end{pgfonlayer}
\end{tikzpicture}\,}

\newkeycommand{\typedkpointdag}[style=kpoint adjoint,edgestyle=][2]{%
\,\begin{tikzpicture}
  \begin{pgfonlayer}{nodelayer}
    \node [style=none] (0) at (0, -1) {};
    \node [style=\commandkey{style}] (1) at (0, 0.75) {$#2$};
    \node [style=right label] (2) at (0.25, -0.5) {$#1$};
  \end{pgfonlayer}
  \begin{pgfonlayer}{edgelayer}
    \draw [\commandkey{edgestyle}] (0.center) to (1);
  \end{pgfonlayer}
\end{tikzpicture}\,}

\newkeycommand{\bistate}[style=kpoint][1]{\,\begin{tikzpicture}
  \begin{pgfonlayer}{nodelayer}
    \node [style=\commandkey{style}, minimum width=1 cm, inner sep=2pt] (0) at (0, -0.25) {$#1$};
    \node [style=none] (1) at (-0.75, 0) {};
    \node [style=none] (2) at (0.75, 0) {};
    \node [style=none] (3) at (-0.75, 0.88) {};
    \node [style=none] (4) at (0.75, 0.88) {};
  \end{pgfonlayer}
  \begin{pgfonlayer}{edgelayer}
    \draw (1.center) to (3.center);
    \draw (2.center) to (4.center);
  \end{pgfonlayer}
\end{tikzpicture}\,}

\newkeycommand{\bistatebig}[style=kpoint][1]{\,\begin{tikzpicture}
  \begin{pgfonlayer}{nodelayer}
    \node [style=\commandkey{style}, minimum width=, minimum width=1.26 cm, inner sep=2pt] (0) at (0, -0.25) {$#1$};
    \node [style=none] (1) at (-1, 0) {};
    \node [style=none] (2) at (1, 0) {};
    \node [style=none] (3) at (-1, 0.88) {};
    \node [style=none] (4) at (1, 0.88) {};
  \end{pgfonlayer}
  \begin{pgfonlayer}{edgelayer}
    \draw (1.center) to (3.center);
    \draw (2.center) to (4.center);
  \end{pgfonlayer}
\end{tikzpicture}\,}

\newkeycommand{\dbistate}[style=dkpoint][1]{\,\begin{tikzpicture}
  \begin{pgfonlayer}{nodelayer}
    \node [style=\commandkey{style}, minimum width=1 cm, inner sep=2pt] (0) at (0, -0.25) {$#1$};
    \node [style=none] (1) at (-0.75, 0) {};
    \node [style=none] (2) at (0.75, 0) {};
    \node [style=none] (3) at (-0.75, 1.0) {};
    \node [style=none] (4) at (0.75, 1.0) {};
  \end{pgfonlayer}
  \begin{pgfonlayer}{edgelayer}
    \draw [boldedge] (1.center) to (3.center);
    \draw [boldedge] (2.center) to (4.center);
  \end{pgfonlayer}
\end{tikzpicture}\,}

\newkeycommand{\bistatebraket}[style1=kpoint,style2=kpointadj][2]{\,\begin{tikzpicture}
  \begin{pgfonlayer}{nodelayer}
    \node [style=\commandkey{style1}, minimum width=1 cm, inner sep=2pt] (0) at (0, -0.75) {$#1$};
    \node [style=none] (1) at (-0.75, -0.5) {};
    \node [style=none] (2) at (0.75, -0.5) {};
    \node [style=none] (3) at (-0.75, 0.5) {};
    \node [style=none] (4) at (0.75, 0.5) {};
    \node [style=\commandkey{style2}, minimum width=1 cm, inner sep=2pt] (5) at (0, 0.75) {$#2$};
  \end{pgfonlayer}
  \begin{pgfonlayer}{edgelayer}
    \draw (1.center) to (3.center);
    \draw (2.center) to (4.center);
  \end{pgfonlayer}
\end{tikzpicture}\,}

\newkeycommand{\dkpoint}[style=dkpoint][1]{\,\tikz{\node[style=\commandkey{style}] (x) at (0,-0.1) {$#1$};\draw[boldedge] (x)--(0,1);}\,}
\newkeycommand{\dkpointadj}[style=dkpointadj][1]{\,\tikz{\node[style=\commandkey{style}] (x) at (0,0.1) {$#1$};\draw[boldedge] (x)--(0,-1);}\,}
\newkeycommand{\dkpointtrans}[style=dkpointtrans][1]{\,\tikz{\node[style=\commandkey{style}] (x) at (0,0.1) {$#1$};\draw[boldedge] (x)--(0,-1);}\,}
\newkeycommand{\dkpointconj}[style=dkpointconj][1]{\,\tikz{\node[style=\commandkey{style}] (x) at (0,-0.1) {$#1$};\draw[boldedge] (x)--(0,1);}\,}

\newcommand{\trace}{\,\begin{tikzpicture}
  \begin{pgfonlayer}{nodelayer}
    \node [style=none] (0) at (0, -0.60) {};
    \node [style=upground] (1) at (0, 0.40) {};
  \end{pgfonlayer}
  \begin{pgfonlayer}{edgelayer}
    \draw (0.center) to (1);
  \end{pgfonlayer}
\end{tikzpicture}\,}

\newcommand\discard\trace

\newkeycommand{\pointketbra}[style1=copoint,style2=point][2]{\,%
\begin{tikzpicture}
  \begin{pgfonlayer}{nodelayer}
    \node [style=none] (0) at (0, -1.75) {};
    \node [style=\commandkey{style1}] (1) at (0, -0.75) {$#1$};
    \node [style=\commandkey{style2}] (2) at (0, 0.75) {$#2$};
    \node [style=none] (3) at (0, 1.75) {};
  \end{pgfonlayer}
  \begin{pgfonlayer}{edgelayer}
    \draw (0.center) to (1);
    \draw (2) to (3.center);
  \end{pgfonlayer}
\end{tikzpicture}\,}

\newkeycommand{\twocopointketbra}[style1=copoint,style2=copoint,style3=point][3]{\,%
\begin{tikzpicture}
  \begin{pgfonlayer}{nodelayer}
    \node [style=none] (0) at (-0.75, -1.75) {};
    \node [style=none] (1) at (0.75, -1.75) {};
    \node [style=\commandkey{style1}] (2) at (-0.75, -0.75) {$#1$};
    \node [style=\commandkey{style2}] (3) at (0.75, -0.75) {$#2$};
    \node [style=\commandkey{style3}] (4) at (0, 0.75) {$#3$};
    \node [style=none] (5) at (0, 1.75) {};
  \end{pgfonlayer}
  \begin{pgfonlayer}{edgelayer}
    \draw (0.center) to (2);
    \draw (1.center) to (3);
    \draw (4) to (5.center);
  \end{pgfonlayer}
\end{tikzpicture}\,}

\newkeycommand{\twopointketbra}[style1=copoint,style2=point,style3=point][3]{\,%
\begin{tikzpicture}
  \begin{pgfonlayer}{nodelayer}
    \node [style=none] (0) at (0, -1.75) {};
    \node [style=none] (1) at (-0.75, 1.75) {};
    \node [style=\commandkey{style1}] (2) at (0, -0.75) {$#1$};
    \node [style=\commandkey{style2}] (3) at (-0.75, 0.75) {$#2$};
    \node [style=\commandkey{style3}] (4) at (0.75, 0.75) {$#3$};
    \node [style=none] (5) at (0.75, 1.75) {};
  \end{pgfonlayer}
  \begin{pgfonlayer}{edgelayer}
    \draw (0.center) to (2);
    \draw (1.center) to (3);
    \draw (4) to (5.center);
  \end{pgfonlayer}
\end{tikzpicture}\,}

\newkeycommand{\pointbraket}[style1=point,style2=copoint,estyle=][2]{\,%
\begin{tikzpicture}
  \begin{pgfonlayer}{nodelayer}
    \node [style=\commandkey{style1}] (0) at (0, -0.625) {$#1$};
    \node [style=\commandkey{style2}] (1) at (0, 0.625) {$#2$};
  \end{pgfonlayer}
  \begin{pgfonlayer}{edgelayer}
    \draw [\commandkey{estyle}] (0) to (1);
  \end{pgfonlayer}
\end{tikzpicture}\,}

\newkeycommand{\kpointketbra}[style1=kpointdag,style2=kpoint,estyle=][2]{\,%
\begin{tikzpicture}
  \begin{pgfonlayer}{nodelayer}
    \node [style=none] (0) at (0, -1.9) {};
    \node [style=\commandkey{style1}] (1) at (0, -0.95) {$#1$};
    \node [style=\commandkey{style2}] (2) at (0, 0.95) {$#2$};
    \node [style=none] (3) at (0, 1.9) {};
  \end{pgfonlayer}
  \begin{pgfonlayer}{edgelayer}
    \draw [\commandkey{estyle}] (0.center) to (1);
    \draw [\commandkey{estyle}] (2) to (3.center);
  \end{pgfonlayer}
\end{tikzpicture}\,}

\newkeycommand{\kpointmap}[style1=kpoint,style2=map][2]{\,%
\begin{tikzpicture}
  \begin{pgfonlayer}{nodelayer}
    \node [style=\commandkey{style1}] (0) at (0, -1.1) {$#1$};
    \node [style=\commandkey{style2}] (1) at (0, 0.5) {$#2$};
    \node [style=none] (2) at (0, 1.75) {};
  \end{pgfonlayer}
  \begin{pgfonlayer}{edgelayer}
    \draw (0) to (1);
    \draw (1) to (2.center);
  \end{pgfonlayer}
\end{tikzpicture}%
\,}

\newkeycommand{\dkpointmap}[style1=dkpoint,style2=dmap][2]{\,%
\begin{tikzpicture}
  \begin{pgfonlayer}{nodelayer}
    \node [style=\commandkey{style1}] (0) at (0, -1.2) {$#1$};
    \node [style=\commandkey{style2}] (1) at (0, 0.4) {$#2$};
    \node [style=none] (2) at (0, 1.7) {};
  \end{pgfonlayer}
  \begin{pgfonlayer}{edgelayer}
    \draw[style=boldedge] (0) to (1);
    \draw[style=boldedge] (1) to (2.center);
  \end{pgfonlayer}
\end{tikzpicture}%
\,}

\newkeycommand{\dkpointmapx}[style1=dkpoint,style2=dmap][2]{\,%
\begin{tikzpicture}
  \begin{pgfonlayer}{nodelayer}
    \node [style=\commandkey{style1}] (0) at (0, -1.4) {$#1$};
    \node [style=\commandkey{style2}] (1) at (0, 0.6) {$#2$};
    \node [style=none] (2) at (0, 1.9) {};
  \end{pgfonlayer}
  \begin{pgfonlayer}{edgelayer}
    \draw[style=boldedge] (0) to (1);
    \draw[style=boldedge] (1) to (2.center);
  \end{pgfonlayer}
\end{tikzpicture}%
\,}

\newkeycommand{\kpointsandwichmap}[style1=kpoint,style2=map,style3=kpointdag][3]{\,%
\begin{tikzpicture}
  \begin{pgfonlayer}{nodelayer}
    \node [style=\commandkey{style1}] (0) at (0, -1.5) {$#1$};
    \node [style=\commandkey{style2}] (1) at (0, 0) {$#2$};
    \node [style=\commandkey{style3}] (2) at (0, 1.5) {$#3$};
  \end{pgfonlayer}
  \begin{pgfonlayer}{edgelayer}
    \draw (0) to (1);
    \draw (1) to (2);
  \end{pgfonlayer}
\end{tikzpicture}%
}

\newkeycommand{\sandwichtwo}[style1=point,style2=map,style3=map,style4=copoint][4]{\,%
\begin{tikzpicture}
  \begin{pgfonlayer}{nodelayer}
    \node [style=\commandkey{style2}] (0) at (0, -1) {$#2$};
    \node [style=\commandkey{style3}] (1) at (0, 1) {$#3$};
    \node [style=\commandkey{style1}] (2) at (0, -2.6) {$#1$};
    \node [style=\commandkey{style4}] (3) at (0, 2.6) {$#4$};
  \end{pgfonlayer}
  \begin{pgfonlayer}{edgelayer}
    \draw (2) to (0);
    \draw (0) to (1);
    \draw (1) to (3);
  \end{pgfonlayer}
\end{tikzpicture}\,}

\newkeycommand{\longbraket}[style1=point,style2=copoint][2]{\,%
\begin{tikzpicture}
  \begin{pgfonlayer}{nodelayer}
    \node [style=\commandkey{style1}] (0) at (0, -2.6) {$#1$};
    \node [style=\commandkey{style2}] (1) at (0, 2.6) {$#2$};
  \end{pgfonlayer}
  \begin{pgfonlayer}{edgelayer}
    \draw (0) to (1);
  \end{pgfonlayer}
\end{tikzpicture}\,}

\newkeycommand{\onb}[style=point]{\ensuremath{\{\,\tikz{\node[style=\commandkey{style}] (x) at (0,-0.3) {$j$};\draw (x)--(0,0.7);}\,\}}\xspace}

\newkeycommand{\phase}[style=white phase dot][1]{\,\begin{tikzpicture}
    \begin{pgfonlayer}{nodelayer}
        \node [style=none] (0) at (0, 0.6) {};
        \node [style=\commandkey{style}] (2) at (0, -0) {$#1$}; 
        \node [style=none] (3) at (0, -0.6) {};
    \end{pgfonlayer}
    \begin{pgfonlayer}{edgelayer}
        \draw (2) to (0.center);
        \draw (3.center) to (2);
    \end{pgfonlayer}
\end{tikzpicture}\,}

\newkeycommand{\dphase}[style=white phase ddot][1]{\,\begin{tikzpicture}
    \begin{pgfonlayer}{nodelayer}
        \node [style=none] (0) at (0, 1.25) {};
        \node [style=\commandkey{style}] (2) at (0, -0) {$#1$};
        \node [style=none] (3) at (0, -1.25) {};
    \end{pgfonlayer}
    \begin{pgfonlayer}{edgelayer}
        \draw [boldedge] (2) to (0.center);
        \draw [boldedge] (3.center) to (2);
    \end{pgfonlayer}
\end{tikzpicture}\,}

\newkeycommand{\dphasepoint}[style=white phase ddot][1]{\,\begin{tikzpicture}[yshift=-3mm]
  \begin{pgfonlayer}{nodelayer}
    \node [style=none] (0) at (0, 1) {};
    \node [style=\commandkey{style}] (1) at (0, 0) {$#1$};
  \end{pgfonlayer}
  \begin{pgfonlayer}{edgelayer}
    \draw [style=boldedge] (0.center) to (1);
  \end{pgfonlayer}
\end{tikzpicture}\,}

\newkeycommand{\dphasepointgray}[style=gray phase ddot][1]{\,\begin{tikzpicture}[yshift=-3mm]
  \begin{pgfonlayer}{nodelayer}
    \node [style=none] (0) at (0, 1) {};
    \node [style=\commandkey{style}] (1) at (0, 0) {$#1$};
  \end{pgfonlayer}
  \begin{pgfonlayer}{edgelayer}
    \draw [style=boldedge] (0.center) to (1);
  \end{pgfonlayer}
\end{tikzpicture}\,}

\newkeycommand{\dphasecopoint}[style=white phase ddot][1]{\,\begin{tikzpicture}[yshift=3mm,yscale=-1]
  \begin{pgfonlayer}{nodelayer}
    \node [style=none] (0) at (0, 1) {};
    \node [style=\commandkey{style}] (1) at (0, 0) {$#1$};
  \end{pgfonlayer}
  \begin{pgfonlayer}{edgelayer}
    \draw [style=boldedge] (0.center) to (1);
  \end{pgfonlayer}
\end{tikzpicture}\,}

\newkeycommand{\dphasepointsm}[style=white phase ddot][1]{\,\begin{tikzpicture}[yshift=-3mm]
  \begin{pgfonlayer}{nodelayer}
    \node [style=none] (0) at (0, 1) {};
    \node [style=\commandkey{style}] (1) at (0, 0) {\footnotesize\!$#1$\!};
  \end{pgfonlayer}
  \begin{pgfonlayer}{edgelayer}
    \draw [style=boldedge] (0.center) to (1);
  \end{pgfonlayer}
\end{tikzpicture}\,}

\newcommand{\greyphase}[1]{\phase[style=grey phase dot]{#1}}

\newkeycommand{\phasepoint}[style=white phase dot][1]{\,\begin{tikzpicture}[yshift=-3mm]
  \begin{pgfonlayer}{nodelayer}
    \node [style=none] (0) at (0, 1) {};
    \node [style=\commandkey{style}] (1) at (0, 0) {$#1$};
  \end{pgfonlayer}
  \begin{pgfonlayer}{edgelayer}
    \draw (0.center) to (1);
  \end{pgfonlayer}
\end{tikzpicture}\,}

\newkeycommand{\phasecopoint}[style=white phase dot][1]{\,\begin{tikzpicture}[yshift=3mm]
  \begin{pgfonlayer}{nodelayer}
    \node [style=none] (0) at (0, -1) {};
    \node [style=\commandkey{style}] (1) at (0, 0) {$#1$};
  \end{pgfonlayer}
  \begin{pgfonlayer}{edgelayer}
    \draw (0.center) to (1);
  \end{pgfonlayer}
\end{tikzpicture}\,}

\newkeycommand{\kpointbraket}[style1=kpoint,style2=kpoint adjoint,estyle=][2]{\,%
\begin{tikzpicture}
  \begin{pgfonlayer}{nodelayer}
    \node [style=\commandkey{style1}] (0) at (0, -0.735) {$#1$};
    \node [style=\commandkey{style2}] (1) at (0, 0.735) {$#2$};
  \end{pgfonlayer}
  \begin{pgfonlayer}{edgelayer}
    \draw [\commandkey{estyle}] (0) to (1);
  \end{pgfonlayer}
\end{tikzpicture}\,}

\newkeycommand{\kpointbraketx}[style1=kpoint,style2=kpoint adjoint,estyle=][2]{\,%
\begin{tikzpicture}
  \begin{pgfonlayer}{nodelayer}
    \node [style=\commandkey{style1}] (0) at (0, -0.885) {$#1$};
    \node [style=\commandkey{style2}] (1) at (0, 0.885) {$#2$};
  \end{pgfonlayer}
  \begin{pgfonlayer}{edgelayer}
    \draw [\commandkey{estyle}] (0) to (1);
  \end{pgfonlayer}
\end{tikzpicture}\,}

\newcommand{\bra}[1]{\ensuremath{\langle #1 |}}
\newcommand{\ket}[1]{\ensuremath{|  #1 \rangle}}

\tikzstyle{cdiag}=[matrix of math nodes, row sep=3em, column sep=3em, text height=1.5ex, text depth=0.25ex,inner sep=0.5em]
\tikzstyle{arrow above}=[transform canvas={yshift=0.5ex}]
\tikzstyle{arrow below}=[transform canvas={yshift=-0.5ex}]

\tikzstyle{env}=[copoint,regular polygon rotate=0,minimum width=0.2cm, fill=black]

\tikzstyle{probs}=[shape=semicircle,fill=white,draw=black,shape border rotate=180,minimum width=1.2cm]

\tikzstyle{bbindex}=[font=\color{blue}\footnotesize]

\tikzstyle{nudge}=[yshift=0.6mm]

\tikzstyle{every picture}=[baseline=-0.25em,scale=0.5]
\tikzstyle{dotpic}=[] 
\tikzstyle{diredges}=[every to/.style={diredge}]
\tikzstyle{math matrix}=[matrix of math nodes,left delimiter=(,right delimiter=),inner sep=2pt,column sep=1em,row sep=0.5em,nodes={inner sep=0pt},text height=1.5ex, text depth=0.25ex]

\tikzstyle{cnot ctrl}=[fill=black, draw=black, shape=circle, inner sep=0pt, minimum width=1.2mm, tikzit category=circuit]
\tikzstyle{cnot targ}=[fill=white, draw=white, shape=circle, tikzit category=circuit, label={center:$\oplus$}, inner sep=0pt, minimum width=2.1mm, tikzit fill={rgb,255: red,102; green,204; blue,255}, tikzit draw=black]

\tikzstyle{inline text}=[]
\tikzstyle{left label}=[label,anchor=east,xshift=2mm]
\tikzstyle{right label}=[label,anchor=west,xshift=-2mm]

\tikzstyle{braceedge}=[decorate,decoration={brace,amplitude=2mm,raise=-1mm}]
\tikzstyle{small braceedge}=[decorate,decoration={brace,amplitude=1mm,raise=-1mm}]

\tikzstyle{doubled}=[line width=1.6pt] 
\tikzstyle{boldedge}=[doubled,shorten <=-0.17mm,shorten >=-0.17mm]
\tikzstyle{boldedgegray}=[doubled,gray,shorten <=-0.17mm,shorten >=-0.17mm]
\tikzstyle{singleedgegray}=[gray]

\tikzstyle{semidoubled}=[line width=1.4pt] 
\tikzstyle{semiboldedgegray}=[semidoubled,gray,shorten <=-0.17mm,shorten >=-0.17mm]

\tikzstyle{boxedge}=[semiboldedgegray]

\tikzstyle{boldedgedashed}=[very thick,dashed,shorten <=-0.17mm,shorten >=-0.17mm]
\tikzstyle{vboldedgedashed}=[doubled,dashed,shorten <=-0.17mm,shorten >=-0.17mm]
\tikzstyle{left hook arrow}=[left hook-latex]
\tikzstyle{right hook arrow}=[right hook-latex]
\tikzstyle{sembracket}=[line width=0.5pt,shorten <=-0.07mm,shorten >=-0.07mm]

\tikzstyle{causal edge}=[->,thick,gray]
\tikzstyle{causal nondir}=[thick,gray]
\tikzstyle{timeline}=[thick,gray, dashed]

\tikzstyle{gray}=[draw=blue!60!white] 

\tikzstyle{cedge}=[<->,thick,gray!70!white]

\tikzstyle{empty diagram}=[draw=gray!40!white,dashed,shape=rectangle,minimum width=1cm,minimum height=1cm]
\tikzstyle{empty diagram small}=[draw=gray!50!white,dashed,shape=rectangle,minimum width=0.6cm,minimum height=0.5cm]

\tikzstyle{dot}=[inner sep=0.3mm,minimum width=2mm,minimum height=2mm,draw,shape=circle,font=\footnotesize]  
\tikzstyle{Wsquare}=[white dot, shape=regular polygon, rounded corners=0.8 mm, minimum size=3.3 mm, regular polygon sides=3, outer sep=-0.2mm]
\tikzstyle{Wsquareadj}=[white dot, shape=regular polygon, rounded corners=0.8 mm, minimum size=3.3 mm, regular polygon sides=3, outer sep=-0.2mm, regular polygon rotate=180]
\tikzstyle{ddot}=[inner sep=0mm, doubled, minimum width=2.5mm,minimum height=2.5mm,draw,shape=circle]

\tikzstyle{black dot}=[dot,fill=black]
\tikzstyle{white dot}=[dot,fill=white]
\tikzstyle{white Wsquare}=[Wsquare,fill=white,text depth=-0.2mm]
\tikzstyle{white Wsquareadj}=[Wsquareadj,fill=white,text depth=-0.2mm]
\tikzstyle{green dot}=[white dot] 
\tikzstyle{gray dot}=[dot,fill=gray!40!white,text depth=-0.2mm]
\tikzstyle{red dot}=[gray dot] 

\tikzstyle{gn}=[circle,fill=zx_green,draw=black,minimum size=5pt,inner sep=1pt]
\tikzstyle{rn}=[circle,fill=zx_red,font=\color{white},draw=black,minimum size=5pt,inner sep=1pt]
\tikzstyle{Hadamard}=[rectangle,fill=yellow,draw=black,minimum size=8pt,inner sep=0pt]

\tikzstyle{black ddot}=[ddot,fill=black]
\tikzstyle{white ddot}=[ddot,fill=white]
\tikzstyle{gray ddot}=[ddot,fill=gray!40!white]

\tikzstyle{gray edge}=[gray!60!white]

\tikzstyle{small dot}=[inner sep=0.5mm,minimum width=0pt,minimum height=0pt,draw,shape=circle]

\tikzstyle{small black dot}=[small dot,fill=black]
\tikzstyle{small white dot}=[small dot,fill=white]
\tikzstyle{small gray dot}=[small dot,fill=gray!40!white]

\tikzstyle{causal dot}=[inner sep=0.4mm,minimum width=0pt,minimum height=0pt,draw=white,shape=circle,fill=gray!40!white]

\tikzstyle{phase dimensions}=[minimum size=4mm,font=\footnotesize,rectangle,rounded corners=2mm,inner sep=0mm,outer sep=-2mm]
\tikzstyle{dphase dimensions}=[minimum size=5mm,font=\footnotesize,rectangle,rounded corners=2.5mm,inner sep=0.2mm,outer sep=-2mm]

\tikzstyle{white phase dot}=[white dot]

\tikzstyle{white phase ddot}=[ddot,fill=white,dphase dimensions]

\tikzstyle{white rect ddot}=[draw=black,fill=white,doubled,minimum size=5mm,font=\footnotesize,rectangle,rounded corners=2.5mm,inner sep=0.2mm]
\tikzstyle{gray rect ddot}=[draw=black,fill=gray!40!white,doubled,minimum size=6mm,font=\footnotesize,rectangle,rounded corners=3mm]

\tikzstyle{gray phase dot}=[gray dot]

\tikzstyle{gray phase ddot}=[ddot,fill=gray!40!white,dphase dimensions]
\tikzstyle{grey phase dot}=[gray phase dot]
\tikzstyle{grey phase ddot}=[gray phase ddot]

\tikzstyle{small phase dimensions}=[minimum size=4mm,font=\tiny,rectangle,rounded corners=2mm,inner sep=0.2mm,outer sep=-2mm]
\tikzstyle{small dphase dimensions}=[minimum size=4mm,font=\tiny,rectangle,rounded corners=2mm,inner sep=0.2mm,outer sep=-2mm]

\tikzstyle{small gray phase dot}=[dot,fill=gray!40!white,small phase dimensions]
\tikzstyle{small gray phase ddot}=[ddot,fill=gray!40!white,small dphase dimensions]

\tikzstyle{Z phase dot}=[minimum size=5mm, font={\footnotesize\boldmath}, shape=rectangle, rounded corners=2mm, inner sep=0.2mm, outer sep=-2mm, scale=0.8, tikzit shape=circle, draw=black, fill=white, tikzit draw=blue]
\tikzstyle{X dot}=[Z dot, shape=circle, draw=black, fill=zx_red]
\tikzstyle{X phase dot}=[Z phase dot, tikzit shape=circle, tikzit draw=blue, fill=gray!40!white, font={\footnotesize\color{black}\boldmath}]

\tikzstyle{halfscalar}=[star,fill=black,draw=black,minimum size=8pt,inner sep=0pt]

\tikzstyle{small map}=[draw,shape=rectangle,minimum height=4mm,minimum width=4mm,fill=white]

\tikzstyle{cnot}=[fill=white,shape=circle,inner sep=-1.4pt]

\tikzstyle{asym hadamard}=[fill=white,draw,shape=NEbox,inner sep=0.6mm,font=\footnotesize,minimum height=4mm]
\tikzstyle{asym hadamard conj}=[fill=white,draw,shape=NWbox,inner sep=0.6mm,font=\footnotesize,minimum height=4mm]
\tikzstyle{asym hadamard dag}=[fill=white,draw,shape=SEbox,inner sep=0.6mm,font=\footnotesize,minimum height=4mm]

\tikzstyle{hadamard}=[fill=white,draw,inner sep=0.6mm,minimum height=1.5mm,minimum width=1.5mm]
\tikzstyle{ehadamard}=[fill=white,draw,double,double distance = 1pt,inner sep=0.9mm,minimum height=1.5mm,minimum width=1.5mm] 
\tikzstyle{small hadamard}=[fill=white,draw,inner sep=0.6mm,minimum height=1.5mm,minimum width=1.5mm]
\tikzstyle{small hadamard rotate}=[small hadamard,rotate=45]
\tikzstyle{dhadamard}=[hadamard,doubled]
\tikzstyle{small dhadamard}=[small hadamard,doubled]
\tikzstyle{small dhadamard rotate}=[small hadamard rotate,doubled]
\tikzstyle{antipode}=[white dot,inner sep=0.3mm,font=\footnotesize]

\tikzstyle{scalar}=[diamond,draw,inner sep=0.5pt,font=\small]
\tikzstyle{dscalar}=[diamond,doubled, draw,inner sep=0.5pt,font=\small]

\tikzstyle{small box}=[rectangle,inline text,fill=white,draw,minimum height=1mm,yshift=-0.3mm,minimum width=4mm,font=\small]
\tikzstyle{small gray box}=[small box,fill=gray!30]
\tikzstyle{medium box}=[rectangle,inline text,fill=white,draw,minimum height=5mm,yshift=-0.5mm,minimum width=8mm,font=\small]
\tikzstyle{square box}=[small box] 
\tikzstyle{medium gray box}=[semilarge box,fill=gray!30,text height=1.5ex, text depth=0.25ex,yshift=0.5mm]
\tikzstyle{semilarge box}=[rectangle,inline text,fill=white,draw,minimum height=5mm,yshift=-0.5mm,minimum width=12.5mm,font=\small]
\tikzstyle{large box}=[rectangle,inline text,fill=white,draw,minimum height=5mm,yshift=-0.5mm,minimum width=15mm,font=\small]
\tikzstyle{large gray box}=[large box,fill=gray!30]

\tikzstyle{Bayes box}=[rectangle,fill=black,draw, minimum height=3mm, minimum width=3mm]

\tikzstyle{gray square point}=[small box,fill=gray!50]

\tikzstyle{dphase box white}=[dhadamard]
\tikzstyle{dphase box gray}=[dhadamard,fill=gray!50!white]
\tikzstyle{phase box white}=[hadamard]
\tikzstyle{phase box gray}=[hadamard,fill=gray!50!white]

\tikzstyle{point}=[regular polygon,regular polygon sides=3,draw,scale=0.75,inner sep=-0.5pt,minimum width=9mm,fill=white,regular polygon rotate=180]
\tikzstyle{copoint}=[regular polygon,regular polygon sides=3,draw,scale=0.75,inner sep=-0.5pt,minimum width=9mm,fill=white]
\tikzstyle{dpoint}=[point,doubled]
\tikzstyle{dcopoint}=[copoint,doubled]

\tikzstyle{wide copoint}=[fill=white,draw,shape=isosceles triangle,shape border rotate=90,isosceles triangle stretches=true,inner sep=0pt,minimum width=1.5cm,minimum height=6.12mm]
\tikzstyle{wide point}=[fill=white,draw,shape=isosceles triangle,shape border rotate=-90,isosceles triangle stretches=true,inner sep=0pt,minimum width=1.5cm,minimum height=6.12mm,yshift=-0.0mm]
\tikzstyle{wide point plus}=[fill=white,draw,shape=isosceles triangle,shape border rotate=-90,isosceles triangle stretches=true,inner sep=0pt,minimum width=1.74cm,minimum height=7mm,yshift=-0.0mm]

\tikzstyle{wide dpoint}=[fill=white,doubled,draw,shape=isosceles triangle,shape border rotate=-90,isosceles triangle stretches=true,inner sep=0pt,minimum width=1.5cm,minimum height=6.12mm,yshift=-0.0mm]

\tikzstyle{tinypoint}=[regular polygon,regular polygon sides=3,draw,scale=0.55,inner sep=-0.15pt,minimum width=6mm,fill=white,regular polygon rotate=180] 

\tikzstyle{white point}=[point]
\tikzstyle{white dpoint}=[dpoint]
\tikzstyle{green point}=[white point] 
\tikzstyle{white copoint}=[copoint]
\tikzstyle{gray point}=[point,fill=gray!40!white]
\tikzstyle{gray dpoint}=[gray point,doubled]
\tikzstyle{red point}=[gray point] 
\tikzstyle{gray copoint}=[copoint,fill=gray!40!white]
\tikzstyle{gray dcopoint}=[gray copoint,doubled]

\tikzstyle{white point guide}=[regular polygon,regular polygon sides=3,font=\scriptsize,draw,scale=0.65,inner sep=-0.5pt,minimum width=9mm,fill=white,regular polygon rotate=180]

\tikzstyle{black point}=[point,fill=black,font=\color{white}]
\tikzstyle{black copoint}=[copoint,fill=black,font=\color{white}]

\tikzstyle{tiny gray point}=[tinypoint,fill=gray!40!white]

\tikzstyle{diredge}=[->]
\tikzstyle{ddiredge}=[<->]
\tikzstyle{rdiredge}=[<-]
\tikzstyle{thickdiredge}=[->, very thick]
\tikzstyle{pointer edge}=[->,very thick,gray]
\tikzstyle{pointer edge part}=[very thick,gray]
\tikzstyle{dashed edge}=[dashed]
\tikzstyle{thick dashed edge}=[very thick,dashed]
\tikzstyle{thick gray dashed edge}=[thick dashed edge,gray!40]
\tikzstyle{thick map edge}=[very thick,|->]

\makeatletter
\newcommand{\boxshape}[3]{%
\pgfdeclareshape{#1}{
\inheritsavedanchors[from=rectangle] 
\inheritanchorborder[from=rectangle]
\inheritanchor[from=rectangle]{center}
\inheritanchor[from=rectangle]{north}
\inheritanchor[from=rectangle]{south}
\inheritanchor[from=rectangle]{west}
\inheritanchor[from=rectangle]{east}
\backgroundpath{
\southwest \pgf@xa=\pgf@x \pgf@ya=\pgf@y
\northeast \pgf@xb=\pgf@x \pgf@yb=\pgf@y

\@tempdima=#2
\@tempdimb=#3

\pgfpathmoveto{\pgfpoint{\pgf@xa - 2pt + \@tempdima}{\pgf@ya}}
\pgfpathlineto{\pgfpoint{\pgf@xa - 2pt - \@tempdima}{\pgf@yb}}
\pgfpathlineto{\pgfpoint{\pgf@xb + 2pt + \@tempdimb}{\pgf@yb}}
\pgfpathlineto{\pgfpoint{\pgf@xb + 2pt - \@tempdimb}{\pgf@ya}}
\pgfpathlineto{\pgfpoint{\pgf@xa - 2pt + \@tempdima}{\pgf@ya}}
\pgfpathclose
}
}}

\boxshape{Strapezoid}{2pt}{2pt}
\boxshape{Ntrapezoid}{-2pt}{-2pt}
\makeatother

\tikzstyle{small Strapezoid}=[draw,shape=Strapezoid,inner sep=2pt,minimum height=5mm,fill=white]
\tikzstyle{medium Strapezoid}=[draw,shape=Strapezoid,inner sep=2pt,minimum height=5mm,fill=white,minimum width=6mm]
\tikzstyle{medium Strapezoid r extended}=[draw,shape=Strapezoid,inner sep=2pt,minimum height=5mm,fill=white,minimum width=6.5mm, xshift=0.5mm]
\tikzstyle{large Strapezoid}=[draw,shape=Strapezoid,inner sep=2pt,minimum height=5mm,fill=white,minimum width=10.5mm]

\tikzstyle{small Ntrapezoid}=[draw,shape=Ntrapezoid,inner sep=2pt,minimum height=5mm,fill=white]
\tikzstyle{medium Ntrapezoid}=[draw,shape=Ntrapezoid,inner sep=2pt,minimum height=5mm,fill=white,minimum width=6mm]
\tikzstyle{medium Ntrapezoid r extended}=[draw,shape=Ntrapezoid,inner sep=2pt,minimum height=5mm,fill=white,minimum width=6.5mm, xshift=0.5mm]
\tikzstyle{large Ntrapezoid}=[draw,shape=Ntrapezoid,inner sep=2pt,minimum height=5mm,fill=white,minimum width=10.5mm]

\tikzstyle{small grey Strapezoid}=[draw,shape=Strapezoid,inner sep=2pt,minimum height=5mm,fill=gray!30]
\tikzstyle{medium grey Strapezoid}=[draw,shape=Strapezoid,inner sep=2pt,minimum height=5mm,fill=gray!30,minimum width=6mm]
\tikzstyle{medium grey Strapezoid r extended}=[draw,shape=Strapezoid,inner sep=2pt,minimum height=5mm,fill=gray!30,minimum width=6.5mm, xshift=0.5mm]
\tikzstyle{large grey Strapezoid}=[draw,shape=Strapezoid,inner sep=2pt,minimum height=5mm,fill=gray!30,minimum width=10.5mm]

\tikzstyle{small grey Ntrapezoid}=[draw,shape=Ntrapezoid,inner sep=2pt,minimum height=5mm,fill=gray!30]
\tikzstyle{medium grey Ntrapezoid}=[draw,shape=Ntrapezoid,inner sep=2pt,minimum height=5mm,fill=gray!30,minimum width=6mm]
\tikzstyle{medium grey Ntrapezoid r extended}=[draw,shape=Ntrapezoid,inner sep=2pt,minimum height=5mm,fill=gray!30,minimum width=6.5mm, xshift=0.5mm]
\tikzstyle{large grey Ntrapezoid}=[draw,shape=Ntrapezoid,inner sep=2pt,minimum height=5mm,fill=gray!30,minimum width=10.5mm]

\makeatletter

\boxshape{NEbox}{0pt}{5pt}
\boxshape{SEbox}{0pt}{-5pt}
\boxshape{NWbox}{5pt}{0pt}
\boxshape{SWbox}{-5pt}{0pt}
\boxshape{EBox}{-3pt}{3pt}
\boxshape{WBox}{3pt}{-3pt}
\makeatother

\tikzstyle{cloud}=[shape=cloud,draw,minimum width=1.5cm,minimum height=1.5cm]

\tikzstyle{map}=[draw,shape=NEbox,inner sep=2pt,minimum height=6mm,fill=white]
\tikzstyle{dashedmap}=[draw,dashed,shape=NEbox,inner sep=2pt,minimum height=6mm,fill=white]
\tikzstyle{mapdag}=[draw,shape=SEbox,inner sep=2pt,minimum height=6mm,fill=white]
\tikzstyle{mapadj}=[draw,shape=SEbox,inner sep=2pt,minimum height=6mm,fill=white]
\tikzstyle{maptrans}=[draw,shape=SWbox,inner sep=2pt,minimum height=6mm,fill=white]
\tikzstyle{mapconj}=[draw,shape=NWbox,inner sep=2pt,minimum height=6mm,fill=white]

\tikzstyle{medium map}=[draw,shape=NEbox,inner sep=2pt,minimum height=6mm,fill=white,minimum width=7mm]
\tikzstyle{medium map dag}=[draw,shape=SEbox,inner sep=2pt,minimum height=6mm,fill=white,minimum width=7mm]
\tikzstyle{medium map adj}=[draw,shape=SEbox,inner sep=2pt,minimum height=6mm,fill=white,minimum width=7mm]
\tikzstyle{medium map trans}=[draw,shape=SWbox,inner sep=2pt,minimum height=6mm,fill=white,minimum width=7mm]
\tikzstyle{medium map conj}=[draw,shape=NWbox,inner sep=2pt,minimum height=6mm,fill=white,minimum width=7mm]
\tikzstyle{semilarge map}=[draw,shape=NEbox,inner sep=2pt,minimum height=6mm,fill=white,minimum width=9.5mm]
\tikzstyle{semilarge map trans}=[draw,shape=SWbox,inner sep=2pt,minimum height=6mm,fill=white,minimum width=9.5mm]
\tikzstyle{semilarge map adj}=[draw,shape=SEbox,inner sep=2pt,minimum height=6mm,fill=white,minimum width=9.5mm]
\tikzstyle{semilarge map dag}=[draw,shape=SEbox,inner sep=2pt,minimum height=6mm,fill=white,minimum width=9.5mm]
\tikzstyle{semilarge map conj}=[draw,shape=NWbox,inner sep=2pt,minimum height=6mm,fill=white,minimum width=9.5mm]
\tikzstyle{large map}=[draw,shape=NEbox,inner sep=2pt,minimum height=6mm,fill=white,minimum width=12mm]
\tikzstyle{large map conj}=[draw,shape=NWbox,inner sep=2pt,minimum height=6mm,fill=white,minimum width=12mm]
\tikzstyle{very large map}=[draw,shape=NEbox,inner sep=2pt,minimum height=6mm,fill=white,minimum width=17mm]

\tikzstyle{medium dmap}=[draw,doubled,shape=NEbox,inner sep=2pt,minimum height=6mm,fill=white,minimum width=7mm]
\tikzstyle{medium dmap dag}=[draw,doubled,shape=SEbox,inner sep=2pt,minimum height=6mm,fill=white,minimum width=7mm]
\tikzstyle{medium dmap adj}=[draw,doubled,shape=SEbox,inner sep=2pt,minimum height=6mm,fill=white,minimum width=7mm]
\tikzstyle{medium dmap trans}=[draw,doubled,shape=SWbox,inner sep=2pt,minimum height=6mm,fill=white,minimum width=7mm]
\tikzstyle{medium dmap conj}=[draw,doubled,shape=NWbox,inner sep=2pt,minimum height=6mm,fill=white,minimum width=7mm]
\tikzstyle{semilarge dmap}=[draw,doubled,shape=NEbox,inner sep=2pt,minimum height=6mm,fill=white,minimum width=9.5mm]
\tikzstyle{semilarge dmap trans}=[draw,doubled,shape=SWbox,inner sep=2pt,minimum height=6mm,fill=white,minimum width=9.5mm]
\tikzstyle{semilarge dmap adj}=[draw,doubled,shape=SEbox,inner sep=2pt,minimum height=6mm,fill=white,minimum width=9.5mm]
\tikzstyle{semilarge dmap dag}=[draw,doubled,shape=SEbox,inner sep=2pt,minimum height=6mm,fill=white,minimum width=9.5mm]
\tikzstyle{semilarge dmap conj}=[draw,doubled,shape=NWbox,inner sep=2pt,minimum height=6mm,fill=white,minimum width=9.5mm]
\tikzstyle{large dmap}=[draw,doubled,shape=NEbox,inner sep=2pt,minimum height=6mm,fill=white,minimum width=12mm]
\tikzstyle{large dmap conj}=[draw,doubled,shape=NWbox,inner sep=2pt,minimum height=6mm,fill=white,minimum width=12mm]
\tikzstyle{large dmap trans}=[draw,doubled,shape=SWbox,inner sep=2pt,minimum height=6mm,fill=white,minimum width=12mm]
\tikzstyle{large dmap adj}=[draw,doubled,shape=SEbox,inner sep=2pt,minimum height=6mm,fill=white,minimum width=12mm]
\tikzstyle{large dmap dag}=[draw,doubled,shape=SEbox,inner sep=2pt,minimum height=6mm,fill=white,minimum width=12mm]
\tikzstyle{very large dmap}=[draw,doubled,shape=NEbox,inner sep=2pt,minimum height=6mm,fill=white,minimum width=19.5mm]

\tikzstyle{muxbox}=[draw,shape=rectangle,minimum height=3mm,minimum width=3mm,fill=white]
\tikzstyle{dmuxbox}=[muxbox,doubled]

\tikzstyle{box}=[draw,shape=rectangle,inner sep=2pt,minimum height=6mm,minimum width=6mm,fill=white]
\tikzstyle{dbox}=[draw,doubled,shape=rectangle,inner sep=2pt,minimum height=6mm,minimum width=6mm,fill=white]
\tikzstyle{dmap}=[draw,doubled,shape=NEbox,inner sep=2pt,minimum height=6mm,fill=white]
\tikzstyle{dmapdag}=[draw,doubled,shape=SEbox,inner sep=2pt,minimum height=6mm,fill=white]
\tikzstyle{dmapadj}=[draw,doubled,shape=SEbox,inner sep=2pt,minimum height=6mm,fill=white]
\tikzstyle{dmaptrans}=[draw,doubled,shape=SWbox,inner sep=2pt,minimum height=6mm,fill=white]
\tikzstyle{dmapconj}=[draw,doubled,shape=NWbox,inner sep=2pt,minimum height=6mm,fill=white]

\tikzstyle{ddmap}=[draw,doubled,dashed,shape=NEbox,inner sep=2pt,minimum height=6mm,fill=white]
\tikzstyle{ddmapdag}=[draw,doubled,dashed,shape=SEbox,inner sep=2pt,minimum height=6mm,fill=white]
\tikzstyle{ddmapadj}=[draw,doubled,dashed,shape=SEbox,inner sep=2pt,minimum height=6mm,fill=white]
\tikzstyle{ddmaptrans}=[draw,doubled,dashed,shape=SWbox,inner sep=2pt,minimum height=6mm,fill=white]
\tikzstyle{ddmapconj}=[draw,doubled,dashed,shape=NWbox,inner sep=2pt,minimum height=6mm,fill=white]

\boxshape{sNEbox}{0pt}{3pt}
\boxshape{sSEbox}{0pt}{-3pt}
\boxshape{sNWbox}{3pt}{0pt}
\boxshape{sSWbox}{-3pt}{0pt}
\tikzstyle{smap}=[draw,shape=sNEbox,fill=white]
\tikzstyle{smapdag}=[draw,shape=sSEbox,fill=white]
\tikzstyle{smapadj}=[draw,shape=sSEbox,fill=white]
\tikzstyle{smaptrans}=[draw,shape=sSWbox,fill=white]
\tikzstyle{smapconj}=[draw,shape=sNWbox,fill=white]

\tikzstyle{dsmap}=[draw,dashed,shape=sNEbox,fill=white]
\tikzstyle{dsmapdag}=[draw,dashed,shape=sSEbox,fill=white]
\tikzstyle{dsmaptrans}=[draw,dashed,shape=sSWbox,fill=white]
\tikzstyle{dsmapconj}=[draw,dashed,shape=sNWbox,fill=white]

\boxshape{mNEbox}{0pt}{10pt}
\boxshape{mSEbox}{0pt}{-10pt}
\boxshape{mNWbox}{10pt}{0pt}
\boxshape{mSWbox}{-10pt}{0pt}
\tikzstyle{mmap}=[draw,shape=mNEbox]
\tikzstyle{mmapdag}=[draw,shape=mSEbox]
\tikzstyle{mmaptrans}=[draw,shape=mSWbox]
\tikzstyle{mmapconj}=[draw,shape=mNWbox]

\tikzstyle{mmapgray}=[draw,fill=gray!40!white,shape=mNEbox]
\tikzstyle{smapgray}=[draw,fill=gray!40!white,shape=sNEbox]

\makeatletter

\pgfdeclareshape{cornerpoint}{
\inheritsavedanchors[from=rectangle] 
\inheritanchorborder[from=rectangle]
\inheritanchor[from=rectangle]{center}
\inheritanchor[from=rectangle]{north}
\inheritanchor[from=rectangle]{south}
\inheritanchor[from=rectangle]{west}
\inheritanchor[from=rectangle]{east}
\backgroundpath{
\southwest \pgf@xa=\pgf@x \pgf@ya=\pgf@y
\northeast \pgf@xb=\pgf@x \pgf@yb=\pgf@y

\pgfmathsetmacro{\pgf@shorten@left}{\pgfkeysvalueof{/tikz/shorten left}}
\pgfmathsetmacro{\pgf@shorten@right}{\pgfkeysvalueof{/tikz/shorten right}}

\pgfpathmoveto{\pgfpoint{0.5 * (\pgf@xa + \pgf@xb)}{\pgf@ya - 5pt}}
\pgfpathlineto{\pgfpoint{\pgf@xa - 8pt + \pgf@shorten@left}{\pgf@yb - 1.5 * \pgf@shorten@left}}
\pgfpathlineto{\pgfpoint{\pgf@xa - 8pt + \pgf@shorten@left}{\pgf@yb}}
\pgfpathlineto{\pgfpoint{\pgf@xb + 8pt - \pgf@shorten@right}{\pgf@yb}}
\pgfpathlineto{\pgfpoint{\pgf@xb + 8pt - \pgf@shorten@right}{\pgf@yb - 1.5 * \pgf@shorten@right}}
\pgfpathclose
}
}

\pgfdeclareshape{cornercopoint}{
\inheritsavedanchors[from=rectangle] 
\inheritanchorborder[from=rectangle]
\inheritanchor[from=rectangle]{center}
\inheritanchor[from=rectangle]{north}
\inheritanchor[from=rectangle]{south}
\inheritanchor[from=rectangle]{west}
\inheritanchor[from=rectangle]{east}
\backgroundpath{
\southwest \pgf@xa=\pgf@x \pgf@ya=\pgf@y
\northeast \pgf@xb=\pgf@x \pgf@yb=\pgf@y

\pgfmathsetmacro{\pgf@shorten@left}{\pgfkeysvalueof{/tikz/shorten left}}
\pgfmathsetmacro{\pgf@shorten@right}{\pgfkeysvalueof{/tikz/shorten right}}

\pgfpathmoveto{\pgfpoint{0.5 * (\pgf@xa + \pgf@xb)}{\pgf@yb + 5pt}}
\pgfpathlineto{\pgfpoint{\pgf@xa - 8pt + \pgf@shorten@left}{\pgf@ya + 1.5 * \pgf@shorten@left}}
\pgfpathlineto{\pgfpoint{\pgf@xa - 8pt + \pgf@shorten@left}{\pgf@ya}}
\pgfpathlineto{\pgfpoint{\pgf@xb + 8pt - \pgf@shorten@right}{\pgf@ya}}
\pgfpathlineto{\pgfpoint{\pgf@xb + 8pt - \pgf@shorten@right}{\pgf@ya + 1.5 * \pgf@shorten@right}}
\pgfpathclose
}
}

\makeatother

\pgfkeyssetvalue{/tikz/shorten left}{0pt}
\pgfkeyssetvalue{/tikz/shorten right}{0pt}

\tikzstyle{kpoint common}=[draw,fill=white,inner sep=1pt,minimum height=4mm]
\tikzstyle{kpoint sc}=[shape=cornerpoint,kpoint common]
\tikzstyle{kpoint adjoint sc}=[shape=cornercopoint,kpoint common]
\tikzstyle{kpoint}=[shape=cornerpoint,shorten left=5pt,kpoint common]
\tikzstyle{kpoint adjoint}=[shape=cornercopoint,shorten left=5pt,kpoint common]
\tikzstyle{kpoint conjugate}=[shape=cornerpoint,shorten right=5pt,kpoint common]
\tikzstyle{kpoint transpose}=[shape=cornercopoint,shorten right=5pt,kpoint common]
\tikzstyle{kpoint symm}=[shape=cornerpoint,shorten left=5pt,shorten right=5pt,kpoint common]

\tikzstyle{black kpoint}=[shape=cornerpoint,shorten left=5pt,kpoint common,fill=black,font=\color{white}]
\tikzstyle{black kpoint adjoint}=[shape=cornercopoint,shorten left=5pt,kpoint common,fill=black,font=\color{white}]
\tikzstyle{black kpointadj}=[shape=cornercopoint,shorten left=5pt,kpoint common,fill=black,font=\color{white}]

\tikzstyle{black dkpoint}=[shape=cornerpoint,shorten left=5pt,kpoint common,fill=black, doubled,font=\color{white}]
\tikzstyle{black dkpoint adjoint}=[shape=cornercopoint,shorten left=5pt,kpoint common,fill=black, doubled,font=\color{white}]
\tikzstyle{black dkpointadj}=[shape=cornercopoint,shorten left=5pt,kpoint common,fill=black, doubled,font=\color{white}] 

\tikzstyle{kpointdag}=[kpoint adjoint]
\tikzstyle{kpointadj}=[kpoint adjoint]
\tikzstyle{kpointconj}=[kpoint conjugate]
\tikzstyle{kpointtrans}=[kpoint transpose]

\tikzstyle{big kpoint}=[kpoint, minimum width=1.2 cm, minimum height=8mm, inner sep=4pt, text depth=3mm]

\tikzstyle{wide kpoint}=[kpoint, minimum width=1 cm, inner sep=2pt]
\tikzstyle{wide kpointdag}=[kpointdag, minimum width=1 cm, inner sep=2pt]
\tikzstyle{wide kpointconj}=[kpointconj, minimum width=1 cm, inner sep=2pt]
\tikzstyle{wide kpointtrans}=[kpointtrans, minimum width=1 cm, inner sep=2pt]

\tikzstyle{gray kpoint}=[kpoint,fill=gray!50!white]
\tikzstyle{gray kpointdag}=[kpointdag,fill=gray!50!white]
\tikzstyle{gray kpointadj}=[kpointadj,fill=gray!50!white]
\tikzstyle{gray kpointconj}=[kpointconj,fill=gray!50!white]
\tikzstyle{gray kpointtrans}=[kpointtrans,fill=gray!50!white]

\tikzstyle{gray dkpoint}=[kpoint,fill=gray!50!white,doubled]
\tikzstyle{gray dkpointdag}=[kpointdag,fill=gray!50!white,doubled]
\tikzstyle{gray dkpointadj}=[kpointadj,fill=gray!50!white,doubled]
\tikzstyle{gray dkpointconj}=[kpointconj,fill=gray!50!white,doubled]
\tikzstyle{gray dkpointtrans}=[kpointtrans,fill=gray!50!white,doubled]

\tikzstyle{white label}=[draw,fill=white,rectangle,inner sep=0.7 mm]
\tikzstyle{gray label}=[draw,fill=gray!50!white,rectangle,inner sep=0.7 mm]
\tikzstyle{black label}=[draw,fill=black,rectangle,inner sep=0.7 mm]

\tikzstyle{dkpoint}=[kpoint,doubled]
\tikzstyle{wide dkpoint}=[wide kpoint,doubled]
\tikzstyle{dkpointdag}=[kpoint adjoint,doubled]
\tikzstyle{wide dkpointdag}=[wide kpointdag,doubled]
\tikzstyle{dkcopoint}=[kpoint adjoint,doubled]
\tikzstyle{dkpointadj}=[kpoint adjoint,doubled]
\tikzstyle{dkpointconj}=[kpoint conjugate,doubled]
\tikzstyle{dkpointtrans}=[kpoint transpose,doubled]

\tikzstyle{kscalar}=[kpoint common, shape=EBox, inner xsep=-1pt, inner ysep=3pt,font=\small]
\tikzstyle{kscalarconj}=[kpoint common, shape=WBox, inner xsep=-1pt, inner ysep=3pt,font=\small]

\tikzstyle{spekpoint}=[kpoint sc,minimum height=5mm,inner sep=3pt]
\tikzstyle{spekcopoint}=[kpoint adjoint sc,minimum height=5mm,inner sep=3pt]

\tikzstyle{dspekpoint}=[spekpoint,doubled]
\tikzstyle{dspekcopoint}=[spekcopoint,doubled]

 \tikzstyle{upground}=[circuit ee IEC,thick,ground,rotate=90,scale=1.5,inner sep=-2mm]
 \tikzstyle{downground}=[circuit ee IEC,thick,ground,rotate=-90,scale=1.5,inner sep=-2mm]

\tikzstyle{maxmix}=[regular polygon,regular polygon sides=3,draw=black,xscale=0.4,yscale=0.3,inner sep=-0.5pt,minimum width=10mm,fill=gray,regular polygon rotate=180]

 \tikzstyle{bigground}=[regular polygon,regular polygon sides=3,draw=gray,scale=0.50,inner sep=-0.5pt,minimum width=10mm,fill=gray]

\tikzstyle{arrs}=[-latex,font=\small,auto]
\tikzstyle{arrow plain}=[arrs]
\tikzstyle{arrow dashed}=[dashed,arrs]
\tikzstyle{arrow bold}=[very thick,arrs]
\tikzstyle{arrow hide}=[draw=white!0,-]
\tikzstyle{arrow reverse}=[latex-]
\tikzstyle{cdnode}=[]

\tikzstyle{Z}=[white dot]
\tikzstyle{wire}=[none]
\tikzstyle{X}=[gray dot]
\tikzstyle{bbox}=[]
\tikzstyle{simple}=[]
\tikzstyle{quanto}=[]

\let\olddagger\dagger
\renewcommand{\dagger}{\ensuremath{\olddagger}\xspace}

\theoremstyle{definition}
\newtheorem{theorem}{Theorem}[section]
\newtheorem{corollary}[theorem]{Corollary}
\newtheorem{lemma}[theorem]{Lemma}
\newtheorem{proposition}[theorem]{Proposition}

\newtheorem{example*}[theorem]{Example*}
\newtheorem{examples*}[theorem]{Examples*}

\newtheorem{remark*}[theorem]{Remark*}

\newtheorem*{theorem*}{Theorem}
\newtheorem*{corollary*}{Corollary}
\newtheorem*{lemma*}{Lemma}
\newtheorem*{proposition*}{Proposition}

\usepackage[color]{changebar}

\usepackage{color}
\def\bR{\begin{color}{red}} 
\def\bB{\begin{color}{blue}}
\def\bM{\begin{color}{magenta}}
\def\bC{\begin{color}{cyan}}
\def\bW{\begin{color}{white}}
\def\bBl{\begin{color}{black}} 
\def\bG{\begin{color}{green}}
\def\bY{\begin{color}{yellow}}
\def\e{\end{color}\xspace}
\newcommand{\bit}{\begin{itemize}}
\newcommand{\eit}{\end{itemize}\par\noindent}
\newcommand{\ben}{\begin{enumerate}}
\newcommand{\een}{\end{enumerate}\par\noindent}
\newcommand{\beq}{\begin{equation}}
\newcommand{\eeq}{\end{equation}\par\noindent}
\newcommand{\beqa}{\begin{eqnarray*}}
\newcommand{\eeqa}{\end{eqnarray*}\par\noindent}
\newcommand{\beqn}{\begin{eqnarray}}
\newcommand{\eeqn}{\end{eqnarray}\par\noindent}

\definecolor{zx_red}{RGB}{227, 145, 145}
\definecolor{zx_green}{RGB}{216, 248, 216}

\title{Graphical Fourier Theory and the Cost of Quantum Addition}

\author{
Stach Kuijpers
\email{se.kuijpers@student.ru.nl}
\and
John van de Wetering
\email{john@vdwetering.name}
\institute{Institute for Computing and Information Sciences\\ Radboud University}
\and
Aleks Kissinger
\email{aleks@cs.ru.nl}
}

\def\titlerunning{Graphical Fourier Theory}
\def\authorrunning{S. Kuijpers, A. Kissinger \& J. van de Wetering}

\begin{document}

\maketitle

\begin{abstract}
The ZX-calculus is a convenient formalism for expressing and reasoning about quantum circuits at a low level, whereas the recently-proposed ZH-calculus yields convenient expressions of mid-level quantum gates such as Toffoli and CCZ. In this paper, we will show that the two calculi are linked by Fourier transform. In particular, we will derive new Fourier expansion rules using the ZH-calculus, and show that we can straightforwardly pass between ZH- and ZX-diagrams using them. Furthermore, we demonstrate that the graphical Fourier expansion of a ZH normal-form corresponds to the standard Fourier transform of a semi-Boolean function. As an illustration of the calculational power of this technique, we then show that several tricks for reducing the T-gate cost of Toffoli circuits, which include for instance quantum adders, can be derived using graphical Fourier theory and straightforwardly generalized to more qubits.
\end{abstract}

\section{Introduction}

The ZX-calculus~\cite{CD1} has proven to be very suitable to reason about quantum computation, when the gate-set is based around the familiar Clifford and phase gates \cite{cliff-simp,chancellor2016graphical,kissinger2019tcount,ZXSurgery}. The ZX-calculus however has no native representation of gates such as the Toffoli or CCZ, and as a result, reasoning about them is quite involved. In contrast, the \emph{ZH-calculus}~\cite{backens2018zh} has generators and rules that allow one to easily reason about these gates and derive identities that would be very non-trivial to realize in the ZX-calculus.

In this paper we develop a bridge between the ZX- and ZH-calculus. We show how a ZH-diagram can be transformed into a ZX-diagram, and back. This translation corresponds to the Fourier transform of semi-Boolean functions, and generalizes the well-known Euler decomposition of the Hadamard gate,
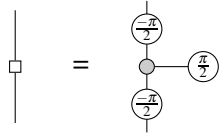
\ctikzfig{hadamard-decomposition}
to CCZ gates and even arbitrary ZH-diagrams. In a sense, it expresses a translation between two different types of multi-qubit phase gates, those being controlled phase gates and exponentiated Pauli gates:
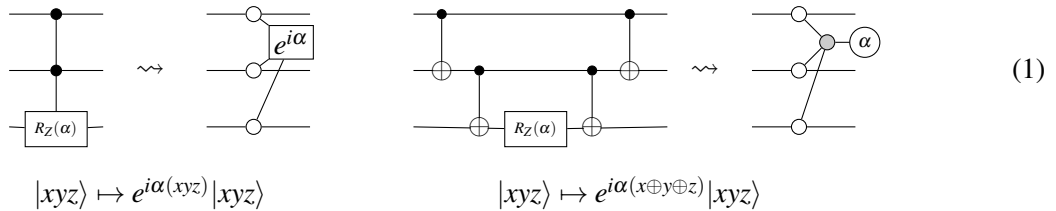
\begin{equation}\label{eq:phase-gate-as-circuit}
\input{phase-gates-comparison.tikz}
\end{equation}
Quantum circuits that are diagonal in the Z-basis can be represented efficiently using \textit{phase polynomials} (see e.g.~\cite{amy2014polynomial,campbelltcount,AmyMoscaReedMuller}). The two gates in \eqref{eq:phase-gate-as-circuit} represent individual terms of this phase polynomial written in two different forms: as a sum of multi-linear or parity terms. The left gate in \eqref{eq:phase-gate-as-circuit}, expressible as a ZH-diagram, introduces a phase if and only if all the inputs are $\ket{1}$. In contrast the right gate, expressible as a ZX-diagram, contains a \emph{phase gadget}~\cite{kissinger2019tcount}. This circuit introduces a phase precisely when the parity of the input state is odd. The Fourier transform allows us to transition between these two different types of phase gates, but is more general in that it allows us to non-trivially transform \emph{any} ZX-diagram into a ZH-diagram and back.

Moreover, the coefficients appearing in our Fourier transform directly correspond to the Fourier spectrum found in the analysis of semi-Boolean functions~\cite{o2014analysis}. A semi-Boolean function is a function $f:\{0,1\}^n\rightarrow\mathbb C$. Any such function can be expanded in either the basis of $\delta$-functions, or the orthogonal basis of parity functions. That is, we can write $f$ either in terms of its $2^n$ output values $\{ \alpha_{\vec{a}} \,|\, \vec{a} \in \mathbb B^n\}$ or its $2^n$ \textit{Fourier coefficients} 
$\{ \widetilde\alpha_{\vec{c}} \,|\, \vec{c} \in \mathbb B^n\}$ as follows:
\[f(\vec{b}) \quad=\quad \sum_{\vec{a}\in\mathbb B^n}\alpha_{\vec{a}}\delta(\vec{a},\vec{b})
\quad=\quad -\frac12 \sum_{\vec{c}\in\mathbb B^n}\widetilde{\alpha}_{\vec{c}}(-1)^{\vec{b}\cdot\vec{c}}\]
where $\vec{b}\cdot\vec{c}$ is the dot-product of Boolean vectors (modulo 2) and $\delta(\vec{a},\vec{b})$ is the Kronecker delta function, which is $1$ if $\vec{a}=\vec{b}$ and $0$ otherwise. Graphically, the Fourier transform can be depicted for any $n$ using the \textit{indexed !-box notation} introduced in Ref.~\cite{backens2018zh} (and reviewed in Section~\ref{s:ZH-dfn}):
\begin{equation}\label{eq:ft-intro}
\input{fourier-transform-intro.tikz}
\qquad\quad
\scalebox{0.8}{$\begin{cases}
\,\ \input{indexing-box.tikz} \; := \; \left(\greyphase{\neg}\right)^{1 - b_1} \ldots \quad \left(\greyphase{\neg}\right)^{1 - b_n} \\[7mm]
\ \input{c-disconnect-box-grey.tikz} \; := \; \left(\begin{tikzpicture}
	\begin{pgfonlayer}{nodelayer}
		\node [style=none] (1) at (-1, 1) {};
		\node [style=none] (2) at (-1, -1) {};
		\node [style=white dot] (13) at (-1, 0.375) {};
		\node [style=gray dot] (14) at (-1, -0.375) {};
	\end{pgfonlayer}
	\begin{pgfonlayer}{edgelayer}
		\draw (1.center) to (13);
		\draw (14) to (2.center);
	\end{pgfonlayer}
\end{tikzpicture}
\right)^{1 - c_1} \ldots \quad \left(\begin{tikzpicture}
	\begin{pgfonlayer}{nodelayer}
		\node [style=none] (1) at (-1, 1) {};
		\node [style=none] (2) at (-1, -1) {};
		\node [style=white dot] (13) at (-1, 0.375) {};
		\node [style=gray dot] (14) at (-1, -0.375) {};
	\end{pgfonlayer}
	\begin{pgfonlayer}{edgelayer}
		\draw (1.center) to (13);
		\draw (14) to (2.center);
	\end{pgfonlayer}
\end{tikzpicture}
\right)^{1 - c_n}
\end{cases}$}
\end{equation}

Using this rule, we can easily switch between the two graphical calculi and exploit the strengths of both. The resulting technique, which we call \textit{graphical Fourier theory}, becomes a powerful new addition to the toolkit for quantum circuit simplification.

It is particularly useful for circuits involving `mid-level' quantum gates such as Toffoli and CCZ, where one often wants to expand to low-level (i.e. Clifford+T) gates while keeping the T count as low as possible. 
This is a topic of active research, as it has strong implications for fault-tolerant implementations of those circuits~\cite{campbell2017roads}, and many \textit{ad hoc} techniques are known.
For example, using an ancilla it is possible to implement a Toffoli gate with 4 T-gates instead of the regular 7~\cite{selinger2013quantum,jones2013low}. Furthermore, a recent paper by Gidney shows how to cut the T-gate cost of a quantum adder circuit in half by simplifying a compute-uncompute pair of Toffoli gates on an ancilla~\cite{gidney2018halving}. Na\"ively, such a configuration would require 8 T-gates, but by means of a clever trick it can actually be implemented using just 4. We show how these tricks and generalisations thereof can be implemented straightforwardly using graphical Fourier theory, suggesting the possibility of more systematic and generalisable simplification techniques.

We recall some basic notions from the ZX- and ZH-calculi in Section~\ref{s:ZH-dfn}, where we will also introduce some new notation. Then in Section~\ref{s:Fourier-transform} we will prove the Fourier transform identities in the ZH-calculus, which we will use in Section~\ref{s:ZX-ZH-transform} to pass between the ZX- and ZH-calculi. Finally, in Section~\ref{s:tcount} we will apply our results to graphically derive T-count reduction techniques for certain kinds of Toffoli circuits.

\section{The ZX- and ZH-calculi}
\label{s:ZH-dfn}

First, we will briefly introduce the ZX- and ZH-calculi. For an extensive introduction to ZX, we refer to Ref.~\cite{CKbook}, Chapter 9, and for ZH we refer to Ref.~\cite{backens2018zh}. Both are diagrammatic languages that represent linear maps between qubits as string diagrams. Their primary point of divergence is in the choice of generators. For ZX, the generators consist of Z-spiders and X-spiders, decorated by phases $\alpha \in [0,2\pi)$:
\[
\input{Z-phase-spider.tikz} := \ket{0...0}\bra{0...0} + e^{i\alpha}\ket{1...1}\bra{1...1} \qquad
\input{X-phase-spider.tikz} := \ket{+...+}\bra{+...+} + e^{i\alpha}\ket{-...-}\bra{-...-}
\]
where $\ket{\pm} := 1/\sqrt{2} \cdot (\ket{0} \pm \ket{1})$ and an omitted phase is assumed to be $0$. Note that we adopt the convention that inputs are at the bottom and outputs are at the top. For ZH, the generators consist of Z-spiders (with phase $0$) and a new ingredient, called an \textit{H-box} that can have a complex parameter $a$:
\[
\input{Z-spider.tikz} := \ket{0\ldots0}\bra{0\ldots0} + \ket{1\ldots1}\bra{1\ldots1} \qquad\qquad
\input{H-spider.tikz} := \sum a^{i_1\ldots i_m j_1\ldots j_n} \ket{j_1\ldots j_n}\bra{i_1\ldots i_m}
\]
where in the right-hand equation, the sum runs over all $i_1,\ldots, i_m, j_1,\ldots, j_n\in\{0,1\}$. An H-box represents a matrix with $a$ as its $\ket{1\ldots1}\bra{1\ldots1}$ entry, and ones elsewhere. By convention, the parameter $a$ is omitted when $a=-1$, so an un-labelled H-box with 1 input and 1 output is a Hadmard gate (up to normalisation).

A pair of diagrams can be composed either by stacking them and joining the outputs of the first with the inputs of the second, which corresponds to composition of linear maps, or placed side-by-side, which represents the tensor product.

For our purposes, we will treat the two spiders of the ZX-calculus as derived generators. To do this, we first define the X-spider and the NOT gate, following Ref.~\cite{backens2018zh}:

\noindent\begin{minipage}{.5\linewidth}
\beq\tag{XS}\label{eq:grey-spider}
 \input{X-spider-dfn.tikz}
\eeq
\end{minipage}
\noindent\begin{minipage}{.49\linewidth}
\beq\tag{N}\label{eq:X-dfn}
 \input{negate-dfn.tikz}
\eeq
\end{minipage}

Note that a grey spider with one output and $m$ inputs evaluated in the computational basis will compute the XOR over its inputs, while the NOT gate computes the negation of a computational basis state.

\begin{figure}[b!]
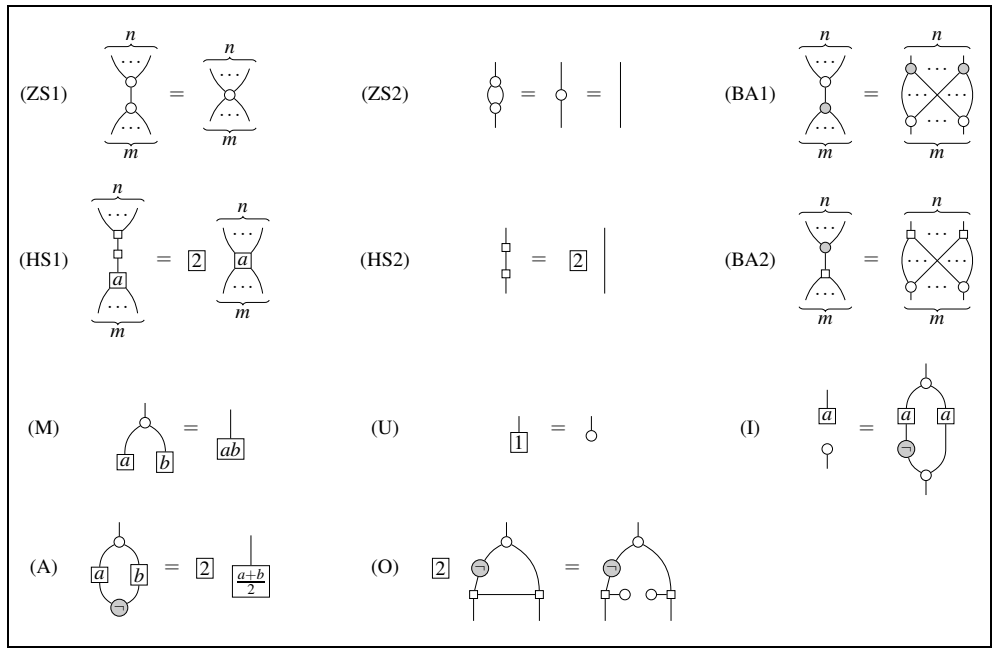

 \centering
 \scalebox{0.7}{
 \begin{tabular}{|cccccccc|}
 \hline 
 &&&&&&&\\
 (ZS1) & \input{Z-spider-rule.tikz} & \qquad \quad & (ZS2) & \input{Z-special.tikz} & \qquad & (BA1) & \input{ZX-bialgebra.tikz} \\ &&&&&&&\\
 (HS1) & \input{H-spider-rule.tikz} & & (HS2) & \input{H-identity.tikz} & & (BA2) & \input{ZH-bialgebra.tikz} \\ &&&&&&&\\
 (M) & \input{multiply-rule.tikz} & & (U) & \input{unit-rule.tikz} & & (I) & \input{intro-rule.tikz} \\ &&&&&&&\\
 (A) & \input{average-rule.tikz} & & (O) & \input{ortho-rule.tikz} & & &\\ &&&&&&&\\
 \hline
 \end{tabular}}
 \caption{The rules of the ZH-calculus.
 Throughout, $m,n$ are nonnegative integers and $a,b$ are arbitrary complex numbers.
 The right-hand sides of both \textit{bialgebra} rules (BA1) and (BA2) are complete bipartite graphs on $(m+n)$ vertices, with an additional input or output for each vertex.}
 \label{fig:ZH-rules}
\end{figure}

Going beyond the original paper we introduce some further derived generators, namely exponentiated H-boxes, which straightforwardly give us the phased spiders of ZX (up to normalisation):
\begin{equation}\label{eq:phase-spider-dfn}
    \input{eH-spider.tikz}  \qquad\qquad \input{X-spider-dfn-phase.tikz} \qquad\qquad \input{Z-spider-dfn-phase.tikz}
\end{equation}
Note that we actually allow $\alpha$ to take any complex value, so these generators are slightly more general than the `usual' ones from the ZX literature.

We consider ZX-diagrams (resp. ZH-diagrams) to be equal when they can be topologically deformed into one another, preserving the order of in- and outputs, or when they can be transformed into another another by using the rules of the ZX-calculus (resp. ZH-calculus). Both the ZX- and ZH-calculi are \textit{complete} in the sense that any equality between matrices can be proven just from diagrammatic rules~\cite{OxfordCompleteness,backens2018zh}. Hence, once we have an encoding of one into the other, the two calculi are equally powerful. Many presentations of the ZX-calculus (and fragments thereof) can be found in the literature (e.g.~\cite{CD1,CD2,backens2016simplified,LoriaCompleteness,OxfordCompleteness,CKbook}), but as we will be referring to rules of the ZH-calculus during our calculations, they are reproduced in Figure~\ref{fig:ZH-rules}. Note that in this paper we will ignore scalar factors.




The rule-set presented in Figure~\ref{fig:ZH-rules} was shown to be complete in Ref.~\cite{backens2018zh}. 
It is furthermore worth noting that the normalisation procedure described in Ref.~\cite{backens2018zh} works as well for diagrams with free parameters (which can be interpreted as matrices of polynomials) as it does for concrete diagrams. We use this fact to conclude that the following two useful equations can be proven in the calculus:

\begin{lemma}\label{lem:NOT-Hbox-negation}\label{lem:copy-AND-id}
A NOT gate can be absorbed into an H-box by negating it's label, and a 3-ary white node and H-box that are doubly connected reduces to a single 2-ary H-box:
\[\input{NOT-Hbox-negation.tikz} \qquad \qquad \input{copy-AND-id.tikz}\]
\end{lemma}


The calculations in this paper will greatly benefit from the use of !-box ("bang box") notation \cite{kissinger2014pattern}. !-boxes represent parts of a diagram that may be replicated an arbitrary number of times, and thus allow one to express a whole family of diagrams at once.
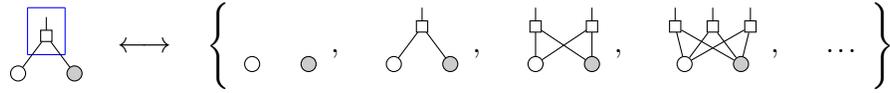
\[
\input{bang-box-example.tikz} \quad \longleftrightarrow \quad
 \left\{
 \ \ \begin{tikzpicture}
	\begin{pgfonlayer}{nodelayer}
		\node [style=white dot] (0) at (-0.75, -0.5) {};
		\node [style=gray dot] (7) at (0.75, -0.5) {};
	\end{pgfonlayer}
\end{tikzpicture}
\ \ ,\quad
 \ \ \input{bang-box-example1.tikz}\ \ ,\quad
 \ \ \input{bang-box-example2.tikz}\ \ ,\quad
 \ \ \input{bang-box-example3.tikz}\ \ ,\quad
 \ \ \ldots\ \  \right\}
\]
When used in equations, corresponding !-boxes on either side of the equation should be replicated an equal number of times. For example, the grey spider \eqref{eq:grey-spider} can be redefined.
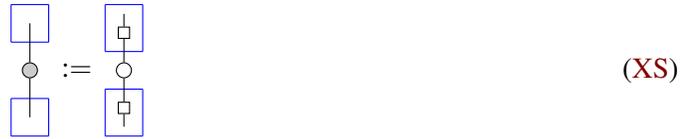
\begin{equation}\tag{\ref{eq:grey-spider}}
    \input{X-spider-dfn-bb.tikz}
\end{equation}

\noindent and one can straightforwardly prove for example that grey spiders fuse just like white spiders (XS1), and that we have a colour changing law (CC) like in the ZX-calculus:

\noindent\begin{minipage}{.5\linewidth}
\beq\tag{XS1}\label{eq:grey-spider-fusion}
 \input{grey-spider-fuse.tikz}
\eeq
\end{minipage}
\noindent\begin{minipage}{.49\linewidth}
\beq\tag{CC}\label{eq:CC}
 \input{color-change.tikz}
\eeq
\end{minipage}

Additionally, like in Ref.~\cite{backens2018zh}, we allow a !-box to be indexed by a finite set $S$. For an indexed !-box, the number of replications is fixed to the size of the set $S$, while the contents of the !-box are parametrized by the elements of the set. For the indexing set we will often use the set of all $n$-bitstrings $\mathbb B^n\equiv \{0,1\}^n$.

The completeness proof of Backens and Kissinger was established by proving that every ZH-diagram can be brought to a unique normal form. Since cups and caps allow us to to change all inputs into outputs, it suffices to consider diagrams with only outputs for this normal-form. The generic structure of the normal-form is:
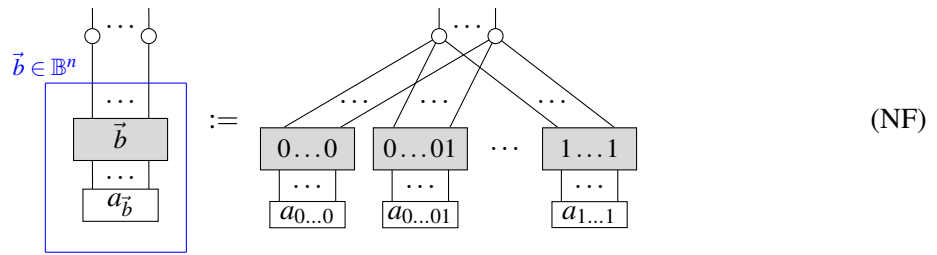
\begin{equation}\tag{NF}\label{eq:NF}
\input{nf-bbox.tikz}\ \ :=\ \ 
\input{nf-picture.tikz}
\end{equation}
The complex numbers on the bottom row correspond to the components of the quantum state $\ket{\vec{a}}$ when expanded in the computational basis: $\ket{\vec{a}}=\sum_{\vec{b}\in\mathbb B^n}a_{\vec{b}} \ket{\vec{b}}$.
The grey box labelled by $\vec{b}$ is the \emph{indexing box} and it is defined as:
\begin{equation}\label{eq:indexing-box-dfn}
\input{indexing-box.tikz} \; = \; \left(\greyphase{\neg}\right)^{1 - b_1} \left(\greyphase{\neg}\right)^{1 - b_2} \ldots \quad \left(\greyphase{\neg}\right)^{1 - b_n}
\end{equation}
The expression $(\ D\ )^x$ for a diagram $D$ and Boolean expression $x$ represents the identity when $x=0$, and $D$ when $x=1$.

\section{The Fourier transform}
\label{s:Fourier-transform}
In this section we will derive our main theorems regarding the Fourier transforms of ZH-diagrams. This theorem will essentially be a broad generalization of the following lemma:
\begin{lemma}\label{lem:2qubit-AND-to-XOR}~
For any $\alpha \in \mathbb{C}$:
\ctikzfig{2qubit-AND-to-XOR-thm}
\end{lemma}
\begin{proof}
Both sides can be brought to normal-form, and easily checked to be equal, i.e.\
 both sides act the same on the four computational basis states \bra{00}, \bra{01}, \bra{10} and \bra{11}.
\end{proof}
Here, the left-hand side is an H-box, while the right-hand side consists of phase gadgets. Note that when $\alpha=\frac{\pi}{2}$ we get the familiar Euler Decomposition of the Hadamard gate. This equality essentially states the identity $\exp(i\alpha (2x_1x_2)) = \exp(i\alpha (x_1+x_2-x_1\oplus x_2))$, for $x_i \in \{0,1\}$.
By inverting this identity, $x_1\oplus x_2=x_1+x_2-2 \ x_1x_2$, we find the inverse to the previous lemma:
\begin{lemma}\label{lem:2qubit-XOR-to-AND}~For any $\alpha \in \mathbb{C}$:
\ctikzfig{2qubit-XOR-to-AND-thm}
\end{lemma}
\begin{proof}
This follows from applying Lemma~\ref{lem:2qubit-AND-to-XOR} to the H-box labelled $-2\alpha$ on the right-hand side and cancelling $\pm\alpha$ phases using (M).
\end{proof}

These two expressions can be generalized to an arbitrary number of qubits and H-boxes/phase gadgets. While these n-qubit forms are powerful rewrite rules on themselves, they will also prove instrumental in our proof of the Fourier transform. In order to show that these two expressions can be generalized to an arbitrary amount of qubits, we need to introduce a new shorthand in the form of \emph{disconnect boxes}. These are diagrams drawn as trapezoids which we label by $\Vec{b}=b_1\dots b_n \in \mathbb B^n$ and have exactly $n$ inputs and $n$ outputs:
\begin{equation}\label{eq:disconnect-box-dfn}
\input{disconnect-box-grey.tikz} \; = \; \left(\begin{tikzpicture}
	\begin{pgfonlayer}{nodelayer}
		\node [style=none] (1) at (-1, 1) {};
		\node [style=none] (2) at (-1, -1) {};
		\node [style=white dot] (13) at (-1, 0.375) {};
		\node [style=gray dot] (14) at (-1, -0.375) {};
	\end{pgfonlayer}
	\begin{pgfonlayer}{edgelayer}
		\draw (1.center) to (13);
		\draw (14) to (2.center);
	\end{pgfonlayer}
\end{tikzpicture}
\right)^{1 - b_1} \ldots \quad \left(\begin{tikzpicture}
	\begin{pgfonlayer}{nodelayer}
		\node [style=none] (1) at (-1, 1) {};
		\node [style=none] (2) at (-1, -1) {};
		\node [style=white dot] (13) at (-1, 0.375) {};
		\node [style=gray dot] (14) at (-1, -0.375) {};
	\end{pgfonlayer}
	\begin{pgfonlayer}{edgelayer}
		\draw (1.center) to (13);
		\draw (14) to (2.center);
	\end{pgfonlayer}
\end{tikzpicture}
\right)^{1 - b_n} \qquad \qquad \input{disconnect-box-white.tikz} \; = \; \left(\input{disconnect-piece-white.tikz}\right)^{1 - b_1} \ldots \quad \left(\input{disconnect-piece-white.tikz}\right)^{1 - b_n}
\end{equation}
Input $i$ is connected by an identity wire to output $i$ as long as $b_i=1$. For $b_i=0$, each of the disconnect boxes will erase a specific connection. The grey trapezoid (dis)connects a white node and grey node, the white trapezoid (dis)connects a white node and H-box. We can use these diagrams to represent a diagram with arbitrary connectivity, since by changing the values in $\vec{b}$, we change the connectivity in the diagram. Note that these disconnect diagrams are not symmetric under vertical reflection, which explains why we draw the diagram asymmetrically. Before we continue, we state some properties that these boxes have, the proofs of which can be found in Appendix~\ref{app:disconnect-boxes}.
\begin{lemma}\label{lem:index-and-disconnect-box-copy}
 The indexing and disconnect boxes copy through white spiders, i.e.\ for any $\vec{b}\in\mathbb B^n$:
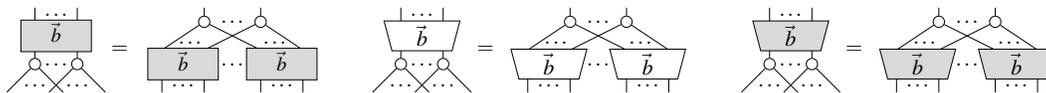
\[ \scalebox{0.75}{\input{indexing-box-copy.tikz} \qquad
\input{disconnect-box-white-copy.tikz} \qquad
\input{disconnect-box-grey-copy.tikz}} \]
\end{lemma}
\noindent Note that this only holds when the 'shorter' end of the trapezoid box is pointing towards the white spiders.

We introduce some special notation for reasoning with bit-strings. When $\vec{c},\vec{d} \in \mathbb{B}^n$ we write $\vec{d}\subseteq \vec{c}$ when $\vec{d}$ only has ones where $\vec{c}$ has ones. I.e.\ if we let $c=101$, then when $d=100$ we have $\vec{d}\subseteq \vec{c}$, but not when $d=010$. We write $\mid \vec{c}\mid$ for the \emph{weight} of the vector, i.e.\ the amount of ones it contains. Finally, we write $\vec{c}\cdot\vec{d}$ to represent their inner product $\sum_i c_id_i$.


\begin{lemma}\label{lem:index-box-eqs}
Disconnect and indexing boxes satisfy the following equations:
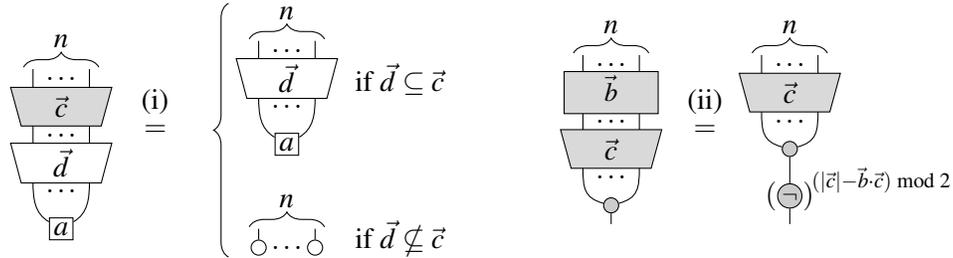
\[
\input{gray-white-disconnect-hbox-erase.tikz} \qquad\qquad
\input{indexing-disconnect-commute.tikz}
\]
\end{lemma}

\begin{lemma}\label{lem:completing-Hbox-connections}
We can substitute an H-box connected to a subset of white spiders by a collection of H-boxes connected to the full set, in the following manner:
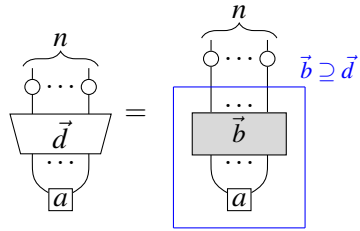
\[\input{completing-Hbox-connections.tikz}\]
\end{lemma}

Using these rules for disconnect boxes we can prove the generalizations of Lemmas~\ref{lem:2qubit-AND-to-XOR} and \ref{lem:2qubit-XOR-to-AND}:

\begin{proposition}\label{lem:AND-to-XOR-rule}
We can expand any exponentiated H-box as a combination of phase-gadgets:
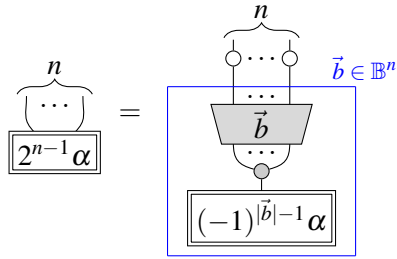
\ctikzfig{AND-to-XOR-thm}
\end{proposition}

\begin{proposition}\label{lem:XOR-to-AND-rule}
We can expand any phase-gadget as a combination of exponentiated H-boxes:
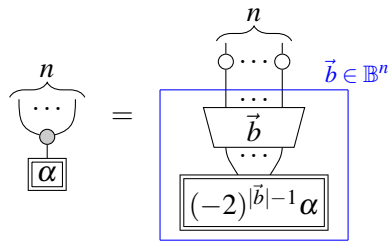
\ctikzfig{XOR-to-AND-thm}
\end{proposition}
These propositions are proven by induction, where the base cases are given by Lemmas~\ref{lem:2qubit-AND-to-XOR} and \ref{lem:2qubit-XOR-to-AND}. The proofs can be found in Appendices \ref{app:AND-to-XOR} and \ref{app:XOR-to-AND}. Our goal now is to extend these propositions to hold for a ZH-diagram normal-form. To do so we will need a few more lemmas.
\begin{lemma}\label{lem:Hbox-same-control-combination}
H-boxes connected to an identical set of white spiders can be combined. In particular, H-boxes connected through an identically labeled disconnect box can be combined.
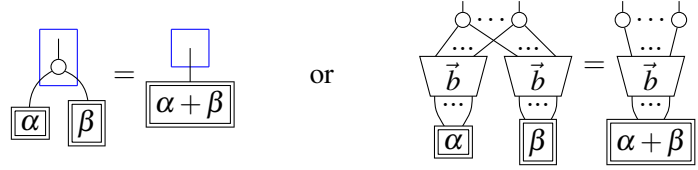
\[\input{Hbox-same-control-combination.tikz}\]
\end{lemma}
\begin{proof}
The left equation is Lemma 2.3 of Ref.~\cite{backens2018zh} and the right follows immediately.
\end{proof}

\begin{lemma}\label{lem:Pgadget-same-control-combination}
Phase gadgets connected to an identical set of white spiders can be combined. In particular, phase gadgets connected through an identically labeled disconnect box can be combined:
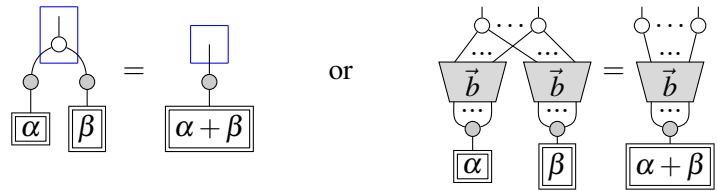
\[\input{Pgadget-same-control-combination.tikz}\]
\end{lemma}
\begin{proof} The left equation follows from the ZH rules (BA1) and (M) and again the right equation is an immediate consequence.
\end{proof}




\noindent We can now state the full version of the graphical Fourier transform:

\begin{theorem}\label{prop:fourier-transform} For all $\vec{b}, \vec{c} \in \mathbb{B}^n$, let $\alpha_{\vec{c}}$ and $\widetilde{\alpha}_{\vec{c}}$ be related according to
\begin{alignat*}{5}
    &\alpha_{\vec{b}} =&& && \sum_{\vec{c}\in \mathbb B^n} \widetilde{\alpha}_{\vec{c}} \ && \Omega (\vec{b},\vec{c})\quad , \qquad && \Omega (\vec{b},\vec{c})=b_1c_1\oplus \dots \oplus b_nc_n \\
    &\widetilde{\alpha}_{\vec{c}} =&& \frac{-1}{2^{n-1}} && \sum_{\vec{b}\in \mathbb B^n} \alpha_{\vec{b}} \ && \chi (\vec{b},\vec{c})\quad , \qquad &&\chi (\vec{b},\vec{c})=(-1)^{\vec{b}\cdot \vec{c}}
\end{alignat*}
then the following equation holds:
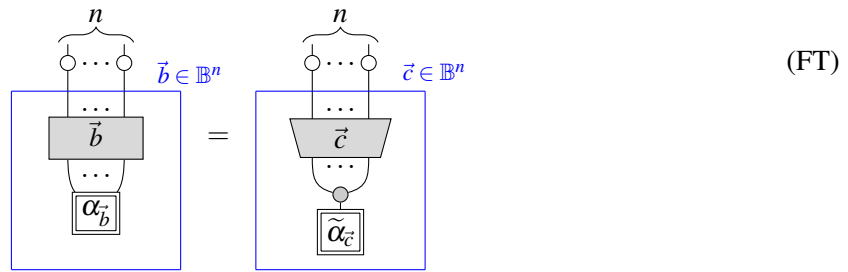
\begin{equation}\tag{FT}\label{eq:fourier-transform}
    \input{Fourier-transform-prop.tikz}    
\end{equation}
\end{theorem}
In particular, when the Fourier coefficients $\widetilde\alpha_{\vec{c}}$ are real, equation~\eqref{eq:ft-intro} from the introduction is recovered. Note that these coefficients differ from the standard coefficients found in the Fourier transform of semi-Boolean functions (e.g. in~\cite{o2014analysis}) by a factor $-2$. This is a consequence of
the phase convention adopted by the ZX-calculus, which affects the behaviour of phase gadgets, up to a global phase. For a more detailed explanation see Appendix~\ref{app:fourier-semi-boolean}.  

\begin{proof}
We will show both the forward direction and the backwards direction. First, from left to right: \\ 
\medskip

{\centering 
\scalebox{0.75}{\input{Fourier-transform-pf1.tikz}} \\}
{\centering 
\scalebox{0.75}{\input{Fourier-transform-pf2.tikz}} \\}

And now from right to left. To do so we need a new operation on bitstrings. We write $\vec{b}*\vec{c}$ to denote the \emph{Schur product}, i.e.\ pairwise AND of $\vec{b}$ and $\vec{c}$. Note that $\vec{d}\subseteq \vec{b},\vec{c}$ if and only if $\vec{d} \subseteq \vec{b}*\vec{c}$. \\
{\centering
\scalebox{0.75}{\input{Inv-Fourier-transform-pf1.tikz}} \\
\scalebox{0.75}{\input{Inv-Fourier-transform-pf2.tikz}} \\
} \medskip
Here we have used that $\sum\limits_{ \vec{d} \subseteq \vec{b}\ast\vec{c}} (-2)^{\mid \vec{d}\mid -1} = \Omega(\vec{b},\vec{c})$ , the proof of which can be found in appendix \ref{app:XOR-sum}.
\end{proof}

\section{Translation between ZX- and ZH-diagrams}
\label{s:ZX-ZH-transform}
The Fourier transform law relates both to the ZX- and ZH-calculus. The left-hand side of \eqref{eq:fourier-transform} resembles the normal-form derived previously by Backens and Kissinger \eqref{eq:NF}, while the right hand side of \eqref{eq:fourier-transform} can be transformed to a ZX-diagram when $\widetilde{\alpha}_{\vec{b}}\in\mathbb{R}$ for all $ \vec{b}\in\mathbb B^n$. Since we used exponentiated H-boxes in the Fourier transform, we cannot express all ZH normal-forms. Nevertheless, we can get arbitrarily close:
\begin{theorem}
Any ZH normal form can be approximated by a Fourier transform.
\end{theorem}
\begin{proof}
Since $\alpha_{\vec{b}}$ and $\widetilde{\alpha}_{\vec{c}}$ in \eqref{eq:fourier-transform} are allowed to be complex numbers, any H-box labeled by $a\in\mathbb C/0$ can be expressed as an exponentiated H-box with label $\alpha=-i\ln{a}$. Having $a=0$ as the only singular value, any H-box can be approximated up to arbitrary precision $\varepsilon$. In particular, every H-box in \eqref{eq:NF} can be approximated, which yields a diagram that is suitable to be transformed according to \eqref{eq:fourier-transform}.

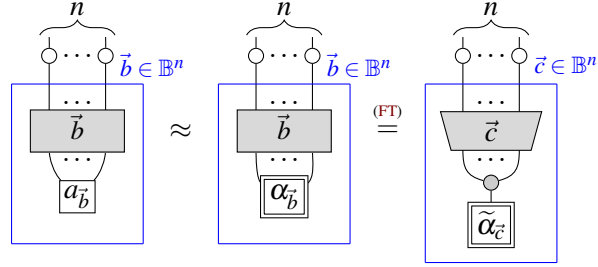
\begin{equation*}
    \input{NF-to-approxNF-to-FT.tikz}
\end{equation*}
\[\text{Where }\ \ \alpha_{\vec{b}} = \begin{cases}
-i\ln{(a_{\vec{b}})} &\text{if $a_{\vec{b}}\neq 0$}\\
-i\ln{(\varepsilon)} &\text{if $a_{\vec{b}} = 0$}
\end{cases}\qquad \text{and}\quad \widetilde{\alpha}_{\vec{c}} = \frac{-1}{2^{n-1}} \sum_{\vec{b}\in \mathbb B^n} \alpha_{\vec{b}} \ \chi (\vec{b},\vec{c}) \qedhere\]
\end{proof}
 
Note that in some cases, this approximation approach can lead to an exact result, like for the cup:
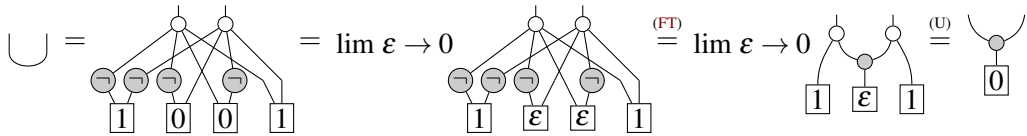
\[\input{cup-fourier-expansion.tikz}\]

\begin{theorem}
A ZH normal-form characterized by a vector $\ket{\vec{a}}$ that only contains phases, i.e.\ $a_{\vec{b}}=e^{i\alpha_{\vec{b}}}$ with $\alpha_{\vec{b}} \in \mathbb R$ for all $ \vec{b}\in\mathbb B^n$, can be expressed as a ZX-diagram by applying a Fourier transform.
\end{theorem}
\begin{proof}
Since $\alpha_{\vec{b}}\in\mathbb{R}$, we also have $\widetilde{\alpha}_{\vec{b}}\in\mathbb R$. One can then write
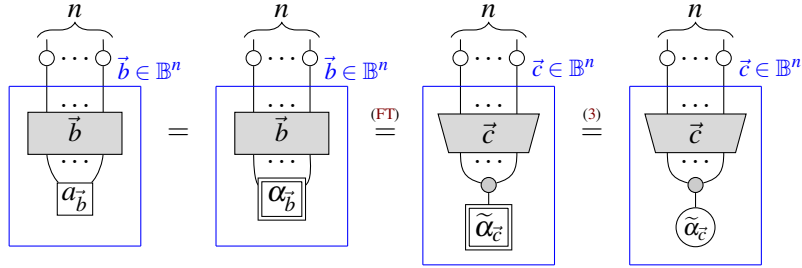
\[\input{ZH-FT-towards-ZX.tikz}\]
Where the final diagram is a ZX-diagram, and
\[a_{\vec{b}} = e^{i\alpha_{\vec{b}}} \qquad \text{with } \  \alpha_{\vec{b}} \in \mathbb R \qquad \text{and} \ \  \widetilde{\alpha}_{\vec{c}} = \frac{-1}{2^{n-1}} \sum_{\vec{b}\in \mathbb B^n} \alpha_{\vec{b}} \ \chi (\vec{b},\vec{c}) \qedhere\]
\end{proof}

As an example, for the standard 3-ary H-box, the coefficient $\alpha_{111}=\pi$ while all other $\alpha_{\vec{b}}$ vanish. Using this we get a 7 T-gate representation of the CCZ gate:
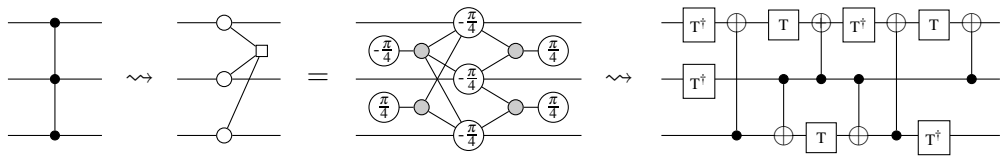
\[\input{CCZ-transform.tikz}\]
The last equality here follows by decomposing the parity phase gates in a smart way, see e.g.\ \cite{amy2014polynomial}.

\section{Reducing non-Clifford count in Toffoli circuits}\label{s:tcount}
In this section we will see how we can use the Fourier transform to reproduce and generalize established results regarding Toffoli circuits and the amount of non-Clifford resources, specifically T-gates, needed to implement them. Since we will be dealing with circuits in this section, we will write our diagrams from left to right, to mimic quantum circuit notation. We will use the translations of phase gates into ZH-diagrams given in equation \eqref{eq:phase-gate-as-circuit}.

Suppose we have a circuit with a Toffoli gate, and a controlled $S^\dagger$ gate, then using the Fourier transform we can derive the following equation:
\begin{equation}\label{eq:had-phase-cancel}
    \input{had-phase-cancel.tikz}
\end{equation}
Whereas naively implementing these gates would require 10 T gates, this cancellation gives the well-known 4 T gate realisation of the Toffoli$^*$ gate~\cite{selinger2013quantum}. We can use this type of cancellation in other settings, by using the following rewrite rule:
\begin{equation}\label{eq:had-phase-rewrite}
    \input{had-phase-rewrite.tikz}
\end{equation}
As we can see, when an X$(\alpha)$-rotation is commuted through an $(n+1)$-ary H-box, it becomes an $n$-controlled Z($\alpha$)-rotation. Similarly to \eqref{eq:had-phase-cancel}, when $\alpha=-\frac{\pi}{2}$, some of the phases present in the Fourier transforms of the gates will cancel.

We can use \eqref{eq:had-phase-cancel} and \eqref{eq:had-phase-rewrite} to derive the 4 T gate plus one ancilla implementation of the Toffoli gate of Ref.~\cite{jones2013low}:
\begin{equation}\label{eq:tof-ancilla}
    \scalebox{0.85}{\input{tof-ancilla.tikz}}
\end{equation}

Using the same technique, this construction generalises to an implementation of an $n$-control Toffoli that uses $4(n-1)$ $Z(\pi/2^n)$-rotation gates, and one ancilla. Unfortunately, in that equation we had post-selected the ancilla to the $\ket{+}$ state. When we actually measure it, and the outcome is $\ket{-}$, we need to apply a correction:
\begin{equation}\label{eq:tof-ancilla-correct}
    \input{tof-ancilla-correct.tikz}
\end{equation}
We see that the measurement outcome is corrected by applying a CZ gate, which is Clifford, and hence not expensive to implement. If we were however implementing a $n$-control Toffoli gate, then the correction would be a $(n-1)$-control Z gate, which will require significant resources to be implemented as well.

Finally, let us derive a version of the trick used in Ref.~\cite{gidney2018halving}, that shows how to implement a compute-uncompute pair of Toffoli's using just 4 T-gates:
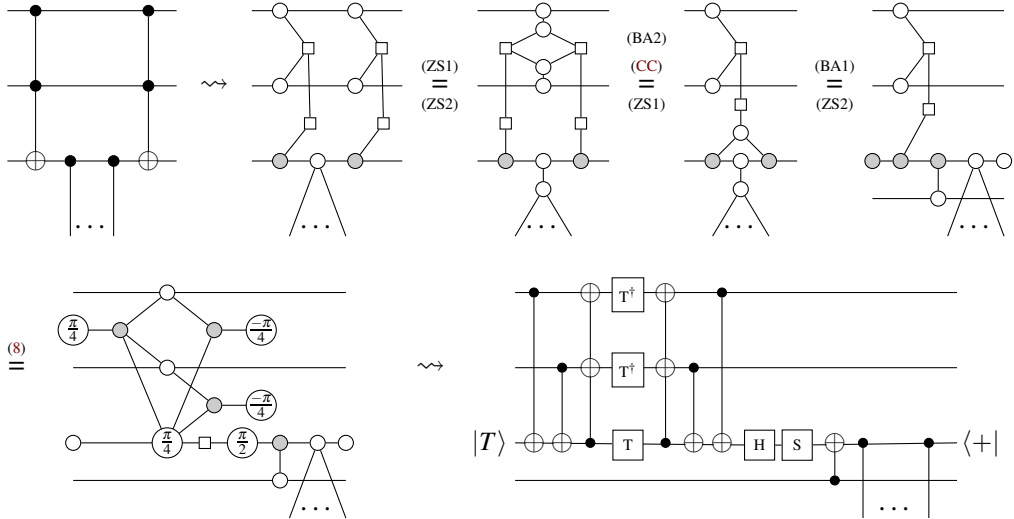
\begin{equation*}
    \input{gidney-rewrite.tikz}
\end{equation*}
Again, this diagram represents a post-selected measurement outcome for the ancillae. Using a similar procedure as in \eqref{eq:tof-ancilla-correct}, the measurement outcome can be corrected by a CZ gate on the top 2 qubits, and a Z gate on the bottom qubit. It should be clear how this procedure generalises to an $n$-control Toffoli compute-uncompute pair. This would again require $4(n-1)$ Z$(\pi/2^n)$-rotation gates to implement, but now the correction becomes a $(n-1)$-control Z gate, which again is non-Clifford when $n>2$.

\section{Conclusion}
We have shown that the Fourier transform of semi-Boolean functions can be applied to ZH-diagrams, leading to a translation between the ZX-calculus and the ZH-calculus. We have then used the Fourier transform to diagrammatically derive and generalise some of the tricks used in Toffoli circuits to reduce the amount of non-Clifford resources used by those circuits.

The next step is to build on the combined ZX- and ZH-diagrammatic reasoning used in Section~\ref{s:tcount} to derive identities for circuits containing Toffoli gates that were not previously known. By doing so we might be able to find circuit simplifications that would be hard, if not impossible, to find using just regular circuit diagrams or even using the ZX-calculus. Furthermore, we aim to extend the particular examples given in this paper to more general patterns, which could be used as automated circuit optimisation strategies, e.g. within the PyZX circuit optimisation tool~\cite{pyzx}.

\noindent \textbf{Acknowledgements.} AK and JvdW are supported in part by AFOSR grant FA2386-18-1-4028.

\newpage

\bibliographystyle{eptcs}
\bibliography{main}

\appendix

\section{The Fourier transform of semi-Boolean functions}\label{app:fourier-semi-boolean}
As shown in the introduction, any semi-Boolean function can be transformed to yield a Fourier spectrum~\cite{o2014analysis}:
\[f(\vec{b})=\alpha_{\vec{b}}= -\frac12 \sum_{c\in\mathbb B^n}\widetilde{\alpha}_{\vec{c}} \chi(\vec{b},\vec{c})\]
where $\chi(\vec{b},\vec{c}):=(-1)^{\vec{b}\cdot\vec{c}}$ is the parity function. To find the coefficients $\widetilde{\alpha}_{\vec{c}}$ one can take the inner product between $\alpha_{\vec{b}}$ and $\chi(\vec{b},\vec{c})$. This works because the parity functions are orthonormal on this inner product:
\[
-\frac12 \widetilde{\alpha}_{\vec{c}} = \langle f(\vec{b}),\chi(\vec{b},\vec{c}) \rangle
\]
hence:
\begin{eqnarray*}
\widetilde{\alpha}_{\vec{c}}
&=& -2 \langle f(\vec{b}),\chi(\vec{b},\vec{c}) \rangle\\
&=& -2 \cdot \frac{1}{2^n}\sum_{\vec{b}\in\mathbb B^n}f(\vec{b})\chi(\vec{b},\vec{c}) \\
&=& \frac{-1}{2^{n-1}}\sum_{\vec{b}\in\mathbb B^n}\alpha_{\vec{b}}\chi(\vec{b},\vec{c})
\end{eqnarray*}


The reason we get a $-1/2$ pre-factor in the Fourier expansion of $f$ comes from the conventions of the ZX-calculus. An $\alpha$-labelled phase gadget introduces a $0$ phase to even paritites of $\vec{c}$ and an $\alpha$ phase to odd-labelled parities. However, this is equal, up to a global phase, to introducing a $-\alpha/2$ phase to even parities and an $\alpha/2$ phase to odd parities. We can show this for the function $f$ by relating $\chi$ to the `other' parity function $\Omega(\vec{b},\vec{c}):= \vec{b} \cdot \vec{c} = b_1c_1\oplus \dots \oplus b_n c_n$ as the parity function of our forward Fourier transform. This modifies the coefficients appearing in the Fourier spectrum. Noting that $\chi(\vec{b},\vec{c})=1-2\Omega (\vec{b},\vec{c})$, we can write:
\begin{eqnarray*}
\alpha_{\vec{b}}
&=& -\frac12 \sum_{\vec{c}\in\mathbb B^n}\widetilde{\alpha}_{\vec{c}} (1-2\Omega (\vec{b},\vec{c})) \\
&=& -\frac12 \sum_{\vec{c}\in\mathbb B^n}\widetilde{\alpha}_{\vec{c}}\ +\sum_{c\in\mathbb B^n}\widetilde{\alpha}_{\vec{c}}\Omega (\vec{b},\vec{c})) \\
&\approx& \sum_{c\in\mathbb B^n}\widetilde{\alpha}_{\vec{c}}\Omega (\vec{b},\vec{c}))
\end{eqnarray*}



\section{Properties of the disconnect boxes}\label{app:disconnect-boxes}


In this appendix we prove the lemmas concerning the indexing and disconnect boxes.



\begin{proof}[Proof of Lemma~\ref{lem:index-and-disconnect-box-copy}.]
The copying of the indexing box is proven in \cite{backens2018zh}. For the disconnect boxes we prove the one qubit case, which straightforwardly generalizes to multiple qubits. We note that when $b=1$, the boxes are the identity so it is trivial. When $b=0$:
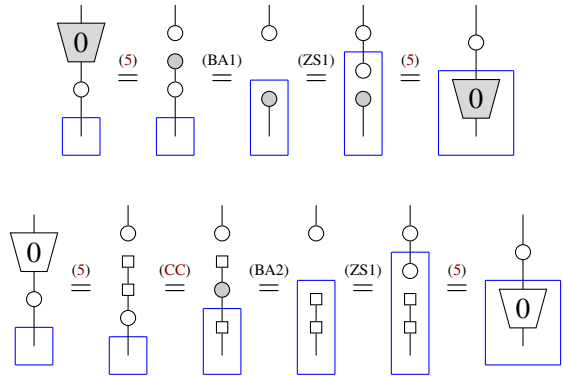
\[\input{disconnect-piece-grey-copy-pf.tikz}\]
\[\input{disconnect-piece-white-copy-pf.tikz}\]
\end{proof}

\begin{proof}[Proof of Lemma~\ref{lem:index-box-eqs}(i)]
For every connection, there are 4 possibilities: $(c_i,d_i)\in\{00,01,10,11\}$. When writing each of these out explicitly side by side, one finds
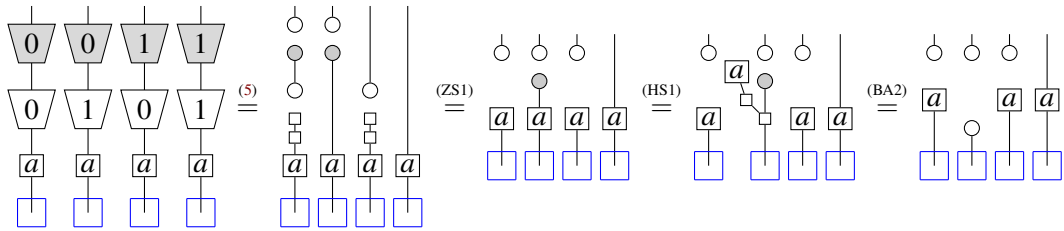
\[\input{gray-white-disconnect-hbox-erase-pf.tikz}\]
We conclude that two disconnect boxes composed with an H-box as depicted in Lemma~\ref{lem:index-box-eqs}(i)
behave like just a single white disconnect box, except when the H-box gets deleted. This occurs when $c_i<d_i$ for some $i$. Equivalently, the H-box is \emph{not} deleted if $\forall \ i \in \{1,\ldots,n\} \ c_i \geq d_i$ and thus $\vec{c} \supseteq\vec{d}$
\end{proof}

\begin{proof}[Proof of Lemma~\ref{lem:index-box-eqs}(ii)]
Again we start by evaluating one connection only. When either $b_i=1$ or $c_i=1$, at least one of the boxes is the identity, and the commutation is trivial. When both $b_i=c_i=0$, the NOT gate gets deleted:
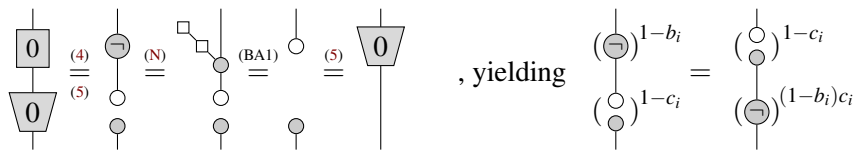
\[\input{indexing-disconnect-commute-pf.tikz}\]
So when there are $n$ wires, $\sum_{i=1}^{n} (1-b_i)c_i=\sum_{i=1}^{n} c_i-b_ic_i= \mid \vec{c}\mid-\vec{b}\cdot\vec{c}\ $ NOT gates remain after commuting the indexing box past the disconnect box. Since two NOT gates cancel, only one NOT gate remains if $(\mid \vec{c}\mid-\vec{b}\cdot\vec{c}) \text{ mod } 2=1$
\end{proof}

\begin{proof}[Proof of Lemma~\ref{lem:completing-Hbox-connections}.]
We prove this by immediately writing the left-hand side in normal-form, and then expanding the indexed !-box:
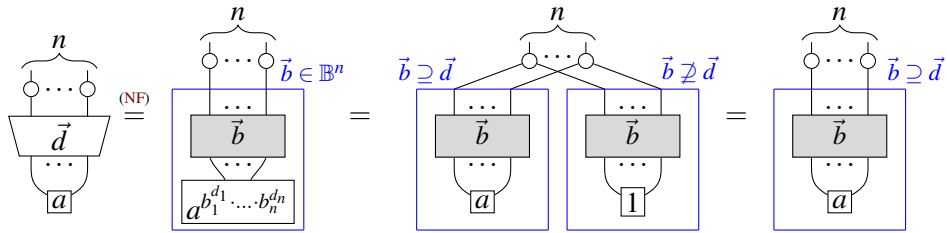
\[\input{completing-Hbox-connections-pf.tikz}\]
\end{proof}

\section{Case distinction}\label{app:case-distinction}
During the proofs presented in appendix \ref{app:AND-to-XOR} and \ref{app:XOR-to-AND}, it will be convenient to explicitly write out part of a labeled bang box, manipulate the result by applying some rules and then recollapse into a single expression. Explicitly, we use that the set of all n-bit strings can be defined inductively as the set of all (n-1)-bit strings appended with either a 0 or a 1.
\[ \mathbb{B}^n = \{\vec{b}b_n \mid \text{ for } \vec{b}\in\mathbb{B}^{n-1} \text{, } b_n\in\{0,1\}=\mathbb{B}^{1} \} \]
Graphically, we can make a similar case distinction on the least significant bit.
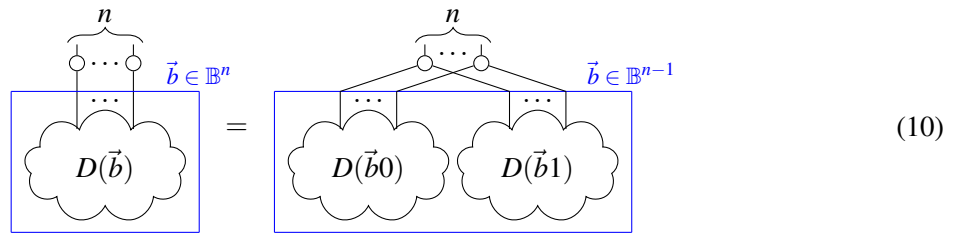
\begin{equation}\label{eq:case-distinction}
    \input{case-distinction-general.tikz}
\end{equation}
\noindent Where the cloud labeled by $D(\vec{b})$ represents any well defined diagram parametrized by $\vec{b}$ in some manner.

Specifically, we will encounter the case where $D(\vec{b})$ contains a disconnect box and an H-box parametrized by $\mid \vec{b} \mid$.

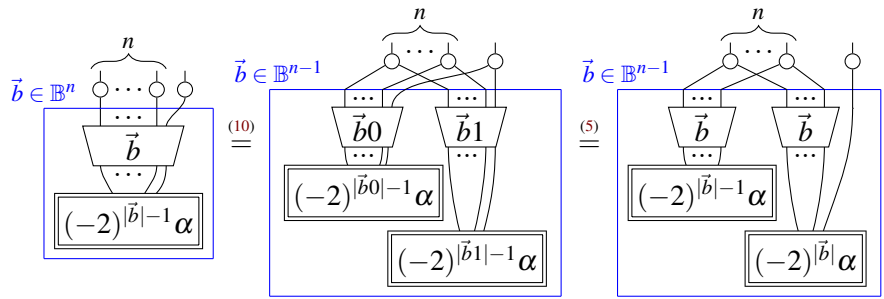
\[\input{case-distinction-example.tikz}\]
Note that we end up with two copies that are distinguished by their connection to one of the white spiders and the power appearing in the H-box. When $D(\vec{b})$ consists of a phase gadget connected to a grey disconnect box, the case distinction is implemented analogously.

\section{Proof of Proposition~\ref{lem:AND-to-XOR-rule}}\label{app:AND-to-XOR}
We will first need the following lemma.:
\begin{lemma}\label{lem:H-box-multiple-legs}
Multiple legs connecting an H-box to a white node are redundant and can be omitted or inserted at will.
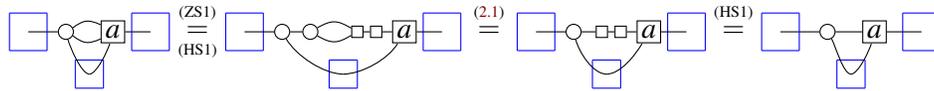
\[\input{Hbox-multiple-legs.tikz}\]
\end{lemma}
\begin{proof}
\[\input{Hbox-multiple-legs-pf.tikz}\]
These steps can be repeated until only one connection is left.
\end{proof}

Now we can actually prove the proposition, namely that:
\ctikzfig{AND-to-XOR-thm}

\begin{proof}
Identifying Lemma~\ref{lem:2qubit-AND-to-XOR} as a base case for n=2, we formulate an inductive proof. The lines of some complete bipartite connections are colored light grey to improve readability . Assuming Lemma \ref{lem:AND-to-XOR-rule} holds for $n$ wires, we can write for $n+1$ wires:

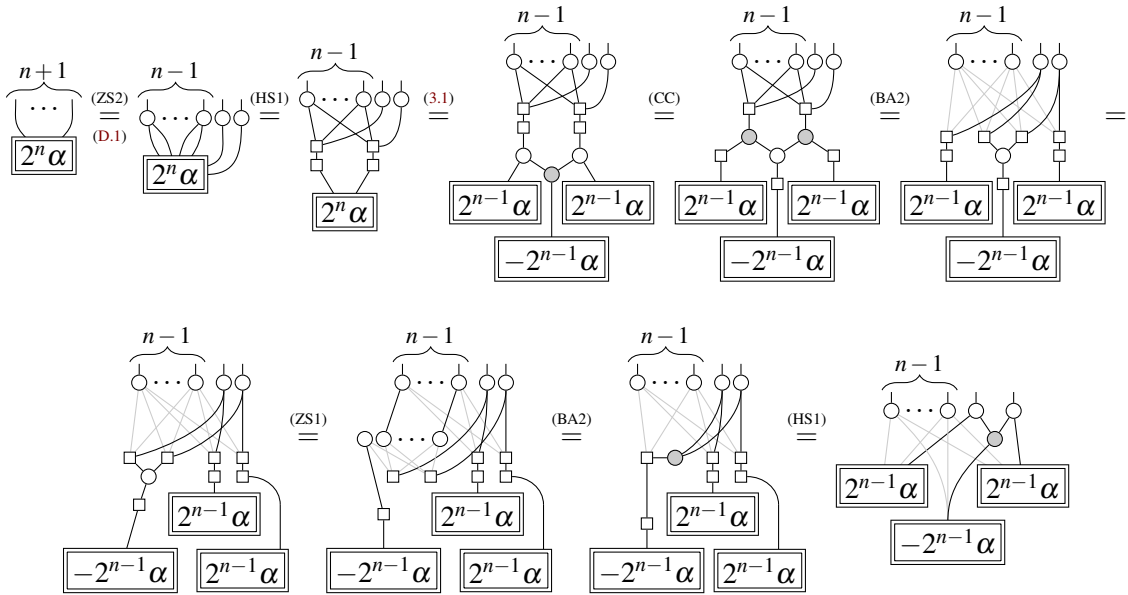
\[\input{AND-to-XOR-pf1.tikz}\]
\[\input{AND-to-XOR-pf2.tikz}\]
where now each of the three $n$-ary H-boxes can be expanded to yield:
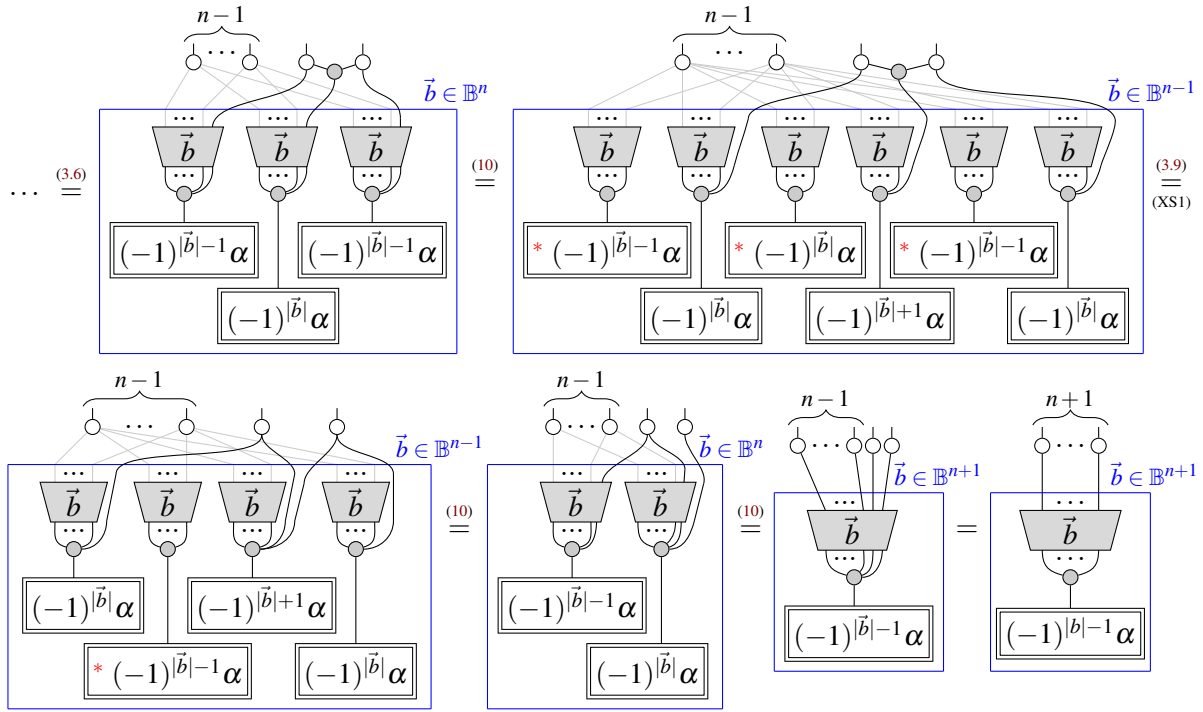
\[\input{AND-to-XOR-pf3.tikz}\]
\[\input{AND-to-XOR-pf4.tikz}\]
which completes the inductive proof. Above, the three phase gadgets labeled by a red asterisk are combined according to Lemma~\ref{lem:Pgadget-same-control-combination} to yield the fourth labeled phase gadget.
\end{proof}

\section{Proof of Proposition~\ref{lem:XOR-to-AND-rule}}\label{app:XOR-to-AND}
We prove that:
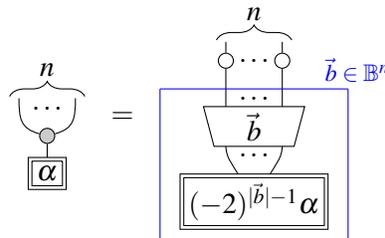
\ctikzfig{XOR-to-AND-thm}

We identify Lemma~\ref{lem:2qubit-XOR-to-AND} as a base case for $n=2$, and we proceed by induction. The lines of some complete bipartite connections are colored light grey to improve readability . Assuming Lemma \ref{lem:XOR-to-AND-rule} holds for $n$ wires, we will write for $n+1$ wires:

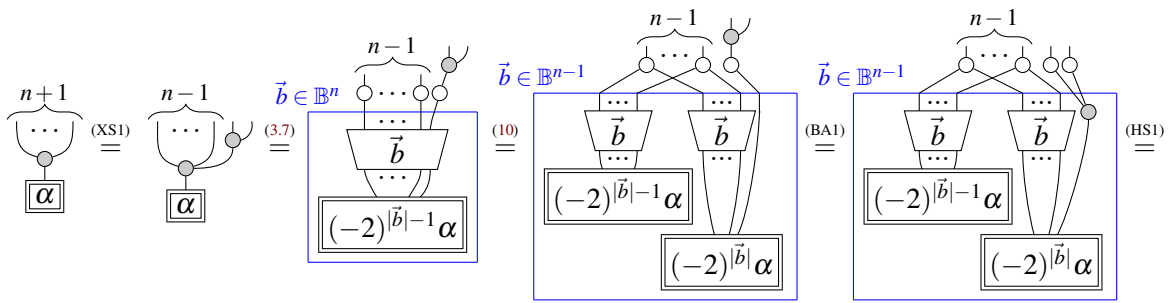
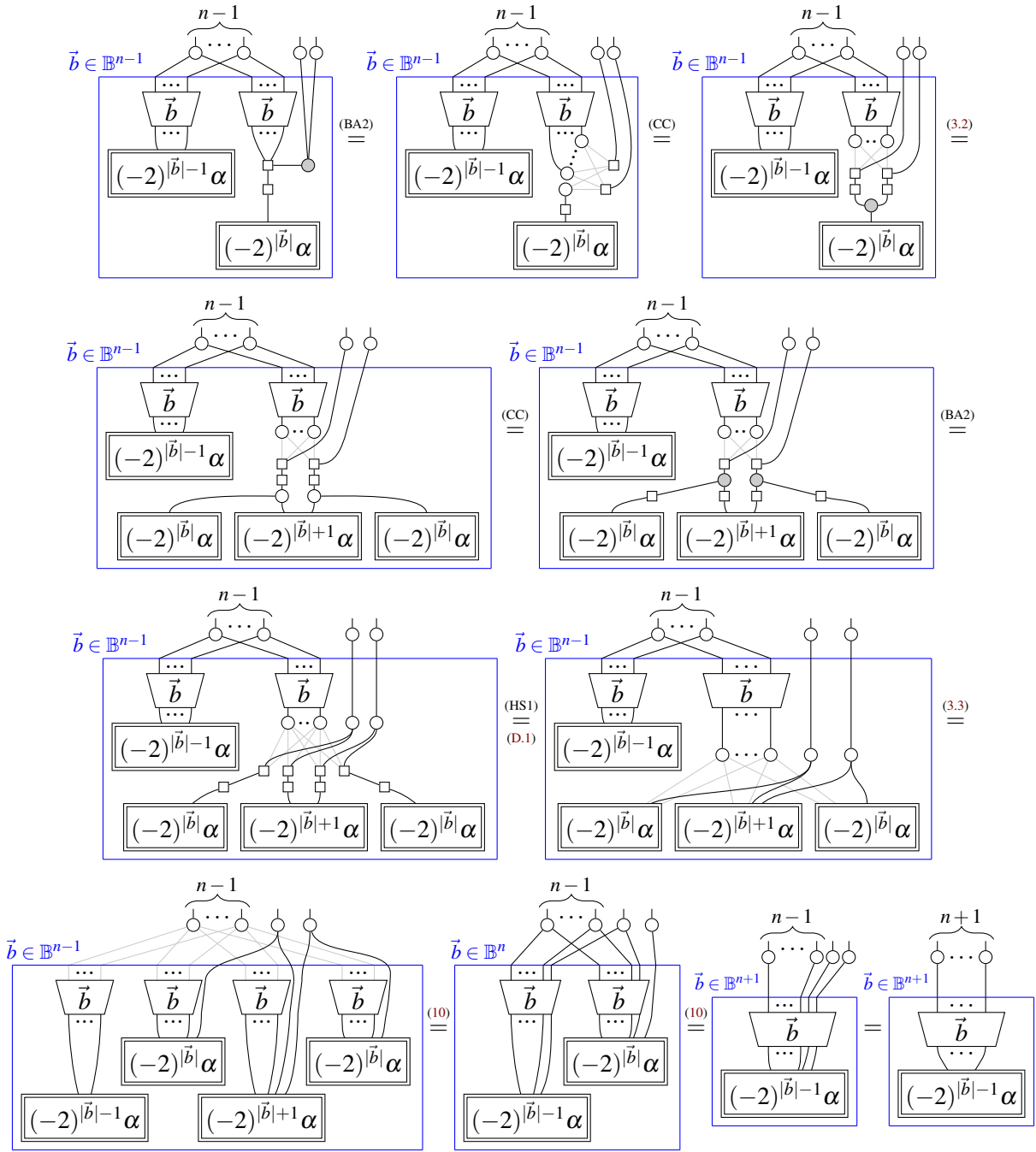
\[\input{XOR-to-AND-pf1.tikz}\]
\[\input{XOR-to-AND-pf2.tikz}\]
\[\input{XOR-to-AND-pf3.tikz}\]
\[\input{XOR-to-AND-pf4.tikz}\]
\[\input{XOR-to-AND-pf5.tikz}\]
This finishes the inductive proof.

\section{XOR as sum over bit-strings}\label{app:XOR-sum}
We set out to prove for $\vec{c} \in \mathbb B^n$
\begin{eqnarray*}
\sum_{ \vec{d} \subseteq \vec{c}} (-2)^{\mid \vec{d}\mid -1} = c_1\oplus \ldots \oplus c_n = (\sum_{i=1}^n c_i)\text{ mod }2 = (\mid \vec{c} \mid)\text{ mod }2
\end{eqnarray*}
\begin{proof}
First note that the sum runs over $\vec{d} \subseteq \vec{c}$, i.e. all bit-strings $\{ \vec{d} \mid d_i \leq c_i \ \forall \ i\in \{1,...,n\}\} / \vec{0}$. We explicitly exclude the zero vector $\vec{0}$ from this set, as $\vec{d}$ is used to parametrize a disconnect box in the proof. For $\vec{d}=\vec{0}$, this corresponds to a completely disconnected H-box, i.e. a scalar. Since these are disregarded throughout, taking them into account here will only lead to inconsistencies.

Next, since the summand is only dependent on the length $\mid\vec{d}\mid$, we count how many of these vectors have the same length. Writing $C=\mid\vec{c}\mid$, $D=\mid\vec{d}\mid$, there are $\binom{C}{D}$ vectors $\vec{d} \subseteq \vec{c}$ with length $D$. Therefore,

\begin{eqnarray*}
\sum_{ \vec{d} \subseteq \vec{c}} (-2)^{\mid \vec{d}\mid -1} &=& \sum_{D=1}^C (-2)^{D-1}\binom{C}{D}
\end{eqnarray*}

We can now inductively prove this is equal to the XOR expression. For $C=1$, the sum consists of only one term and both expressions immediately evaluate to $1$.
For $C+1$, we get for the XOR expression

\begin{equation*}
c_1\oplus \ldots \oplus c_n = (C+1)\text{ mod }2 = 1-(C\text{ mod }2) = 1-\tilde{c}_1\oplus \ldots \oplus \tilde{c}_n
\end{equation*}

\noindent While the sum over bitstrings expression yields

\begin{eqnarray*}
\sum_{D=1}^{C+1} (-2)^{D-1}\binom{C+1}{D}
&=& (-2)^C + \sum_{D=1}^{C} (-2)^{D-1}\binom{C+1}{D}\\
&=&(-2)^C+ \sum_{D=1}^{C} (-2)^{D-1}\left(\binom{C}{D}+\binom{C}{D-1}\right) \\
&=&(-2)^C+ \sum_{D=1}^{C} (-2)^{D-1}\binom{C}{D}+ \sum_{D=1}^{C} (-2)^{D-1}\binom{C}{D-1} \\
&=&(-2)^C+ \sum_{D=1}^{C} (-2)^{D-1}\binom{C}{D}+ \sum_{D=0}^{C-1} (-2)^{D}\binom{C}{D} \\
&=&(-2)^C+ (-2)^{C-1} + \sum_{D=1}^{C-1} (-2)^{D-1}\binom{C}{D} + (-2)^0 + \sum_{D=1}^{C-1} (-2)^{D}\binom{C}{D} \\
&=&(-2)^C+ (-2)^{C-1} + 1 + \sum_{D=1}^{C-1} \left((-2)^{D-1}+(-2)^D\right)\binom{C}{D} \\
&=& 1-(-2)^{C-1} + \sum_{D=1}^{C-1} -(-2)^{D-1}\binom{C}{D} \\
&=& 1-(-2)^{C-1}\binom{C}{C} - \sum_{D=1}^{C-1} (-2)^{D-1}\binom{C}{D} \\
&=& 1-\sum_{D=1}^{C} (-2)^{D-1}\binom{C}{D}
\end{eqnarray*}

Since both show the same behaviour, we can conclude the two representations are instances of the same function.
\end{proof}

\begin{corollary}\label{cor:alternative-XOR}
We can write, for $\vec{b},\vec{c} \in \mathbb B^n$
\begin{eqnarray*}
\sum_{ \vec{d} \subseteq \vec{b}\ast\vec{c}} (-2)^{\mid \vec{d}\mid -1}
&=& (\vec{b}\ast\vec{c})_1\oplus \ldots \oplus (\vec{b}\ast\vec{c})_n \\
&=& \vec{b}_1\vec{c}_1 \oplus \ldots \oplus \vec{b}_n\vec{c}_n \\
&\equiv& \Omega(\vec{b},\vec{c})
\end{eqnarray*}
\end{corollary}

\end{document}